\documentclass[onecolumn]{IEEEtran}
\usepackage{amsfonts,amsmath,amssymb,amsthm,mathrsfs}
\usepackage{graphicx,multirow,bm}
\usepackage{color}
\usepackage[noadjust]{cite}
\makeatletter
\newtheoremstyle{mythm}{3pt}{3pt}{}{16pt}{\bfseries}{:}{.5em}{}
\theoremstyle{mythm}
\newtheorem{theorem}{Theorem}
\newtheorem{example}{Example}
\newtheorem{definition}{Definition}
\newtheorem{remark}{Remark}

\newtheorem{corollary}{Corollary}
\newtheorem{lemma}{Lemma}
\newtheorem{construction}{Construction}

\newcommand{\cA}{\mathcal{A}}
\newcommand{\cB}{\mathcal{B}}
\newcommand{\cC}{\mathcal{C}}

\newcommand{\cE}{\mathcal{E}}

\newcommand{\cG}{\mathcal{G}}

\newcommand{\cR}{\mathcal{R}}
\newcommand{\cS}{\mathcal{S}}

\newcommand{\cV}{\mathcal{V}}

\newcommand{\Z}{\mathbb{Z}}
\newcommand{\F}{\mathbb{F}}

\DeclareMathOperator{\lcm}{lcm}
\DeclareMathOperator{\spn}{Span}
\DeclareMathOperator{\rank}{Rank}

\renewcommand{\le}{\leqslant}
\renewcommand{\leq}{\leqslant}
\renewcommand{\ge}{\geqslant}
\renewcommand{\geq}{\geqslant}

\newcommand{\mathset}[1]{\left\{#1\right\}}
\newcommand{\abs}[1]{\left|#1\right|}
\newcommand{\ceilenv}[1]{\left\lceil #1 \right\rceil}
\newcommand{\floorenv}[1]{\left\lfloor #1 \right\rfloor}

\newcommand{\parenv}[1]{\left( #1 \right)}

\begin{document}

\title{On Optimal Locally Repairable Codes and Generalized Sector-Disk Codes}
\author{
Han Cai,~\IEEEmembership{Member,~IEEE},
and Moshe Schwartz,~\IEEEmembership{Senior Member,~IEEE}
\thanks{The material in this paper was submitted in part to the IEEE International Symposium on Information Theory (ISIT 2020), Los Angeles, CA, USA.}%
\thanks{Han Cai is with the School
   of Electrical and Computer Engineering, Ben-Gurion University of the Negev,
   Beer Sheva 8410501, Israel
   (e-mail: hancai@aliyun.com).}%
\thanks{Moshe Schwartz is with the School
   of Electrical and Computer Engineering, Ben-Gurion University of the Negev,
   Beer Sheva 8410501, Israel
   (e-mail: schwartz@ee.bgu.ac.il).}%
\thanks{This work was supported in part by a German Israeli Project Cooperation (DIP) grant under grant no.~PE2398/1-1.}
}

\maketitle
\begin{abstract}
  Optimal locally repairable codes with information locality are
  considered. Optimal codes are constructed, whose length is also
  order-optimal with respect to a new bound on the code length derived
  in this paper. The length of the constructed codes is super-linear
  in the alphabet size, which improves upon the well known pyramid
  codes, whose length is only linear in the alphabet size.  The
  recoverable erasure patterns are also analyzed for the new codes.
  Based on the recoverable erasure patterns, we construct generalized
  sector-disk (GSD) codes, which can recover from disk erasures mixed
  with sector erasures in a more general setting than known
  sector-disk (SD) codes. Additionally, the number of sectors in the
  constructed GSD codes is super-linear in the alphabet size, compared
  with known SD codes, whose number of sectors is only linear in the
  alphabet size.
\end{abstract}


\section{Introduction}

\IEEEPARstart{I}{n} the large distributed storage systems of today,
disk failures are the norm rather than the exception. Thus,
erasure-coding techniques are employed to protect the data from disk
failures. An $[n,k]$ storage code encodes $k$ information symbols to
$n$ symbols and stores them across $n$ disks in a storage
system. Generally speaking, among all storage codes, maximum distance
separable (MDS) codes are preferred for practical systems both in
terms of redundancy and in terms of reliability. However, as pointed
in \cite{SAPDVCB}, MDS codes such as Reed-Solomon codes suffer from a
high repair cost. This is mainly because, for an $[n,k]$ MDS code,
whenever one wants to recover a symbol, one needs to contact $k$
surviving symbols, which is costly, especially in large-scale
distributed file systems.

To improve the repair efficiency, locally repairable codes, such as
pyramid codes \cite{HCL}, are deployed to reduce the number of symbols
contacted during the repair process. More precisely, the concept of
$r$-locality for a code $\mathcal{C}$ was initially studied in
\cite{GHJY} to ensure that a failed symbol can be recovered by only
accessing $r\ll k$ other symbols which form a repair set.

In the past decade, the original definition has been generalized in
different aspects. Firstly, to guarantee that the system can recover
locally from multiple erasures, the notion of $r$-locality was
generalized to $(r,\delta)$-locality, namely, each repair set is
capable of recovering from $\delta-1$ erasures. Secondly, to let code
symbols have good availability, the notion of locality has been
generalize to $(r,\delta)$-availability \cite{RPDV} (or
$(r,\delta)_c$-locality \cite{WZ}), in which case a code symbol has
more than one repair set. Thus, each repair set can be viewed as a
backup for the target code symbol, hence the code symbol can be
accessed independently through each repair set.  Finally, to satisfy
differing locality requirements, the notion of locality has been
generalized to the hierarchical and the unequal locality cases. Upper
bounds on the minimum Hamming distance of locally repairable codes and
constructions for them have been reported in the literature for those
generalizations. For examples, the reader may refer to
\cite{Bla,BH,CMST_SL,CXHF,CM,HCL,Jin,LMC,LMX,MK,PKLK,RKSV,SDYL,TPD,TB,WFEH}
for $(r,\delta)$-locality, \cite{CCFT,CMST,HX,RPDV,SES,TBF,WZ} for
$(r,\delta)$-availability, \cite{SAK} for hierarchical locality, and
\cite{KL,ZY} for unequal locality.

Based on the observation given in \cite{GHJY}, locally repairable
codes may recover from some special erasure patterns beyond their
minimum Hamming distance. Thus, another research branch for locally
repairable codes is the study of their recoverable erasure patterns.
In this aspect, two special kinds of codes have received most of the
attention. One is the $(\delta-1,\gamma)$-maximally recoverable code
first introduced in \cite{GHJY,BHH}, that can recover from erasure
patterns that include any $\delta-1$ erasures from each repair set,
and any other $\gamma$ erasures. The $(\delta-1,\gamma)$-maximally
recoverable codes are equivalent to $(\delta-1,\gamma)$-partial MDS
codes a special kind of array codes that was introduced to improve the
storage efficiency of redundant arrays of independent disks (RAIDs)
\cite{BHH}.  The other is $(\delta-1,\gamma)$-sector-disk (SD) codes
\cite{PB} that can recover from erasure patterns that include any
$\delta-1$ erasures from each repair set with consistent indices
(i.e., whole disk erasures) and any other $\gamma$ erasures (i.e.,
sector erasures). For construction of SD codes the reader may refer to
\cite{BHH,BPSY,CK,GHJY,GYBS,MK,PB} for example. The main drawback of
all of the reported constructions for SD codes is the requirement for
a large finite field.

In this paper, we focus on both $(r,\delta)$-locality and recoverable
erasure patterns beyond the minimum Hamming distance. For
$(r,\delta)$-locality we propose constructions of locally repairable
codes whose information symbols have $(r,\delta)$-locality and their
length is super-linear in the field size. The codes generated by our
constructions have new parameters compared with known locally
repairable codes. In particular, our codes have a smaller requirement
on the field size compared with pyramid codes. Additionally, we
consider the following fundamental problem: how long can a locally
repairable codes be, whose information symbols have
$(r,\delta)$-locality? We propose a new upper bound on the length of
optimal locally repairable codes. Based on this bound, we prove that
the codes generated by our construction may have order-optimal length.
We also analyze recoverable erasure patterns beyond the minimum
Hamming distance in the codes we construct. Based on this analysis, we
construct array codes that can recover special erasure patterns which
mix whole disk erasures together with additional sector erasures that
beyond the minimum Hamming distance. These codes generalize SD codes,
and we therefore call them generalized sector-disk (GSD)
codes. Finally, the classic Goppa codes are modified into locally
repairable codes. In this way, the generated codes not only share
similar parameters with the ones in \cite{CMST_SL}, but also yield
optimal locally repairable codes with new parameters.

The remainder of this paper is organized as
follows. Section~\ref{sec-preliminaries} introduces some necessary
notation and results.  Section~\ref{sec-LRC} proposes a new
construction of locally repairable codes and a bound on code
length. In Section~\ref{sec-GSD} we introduce GSD codes, and in
Section~\ref{sec-Goppa} we modify classical Goppa codes into a class
of locally repairable codes.  Section~\ref{sec-conclusion} concludes
this paper with some remarks.

\section{Preliminaries}\label{sec-preliminaries}

Throughout this paper,  the following notation are used:
\begin{itemize}
\item For a positive integer $n$, let $[n]$ denote the set
  $\{1,2,\cdots,n\}$;
\item For any prime power $q$, let $\mathbb{F}_q$ denote the finite field with
$q$ elements;
\item An $[n,k]_q$ linear code $\mathcal{C}$ over $\mathbb{F}_q$ is a
  $k$-dimensional subspace of $\mathbb{F}_q^n$ with a $k\times n$
  generator matrix $G=({\bf g}_1,{\bf g}_2,\cdots,{\bf g}_{n})$, where
  ${\bf g}_i$ is a column vector of length $k$ for all $1\le i\le
  n$. Specifically, it is called an $[n,k,d]_q$ linear code if the
  minimum Hamming distance is $d$;
\item For a subset $S\subseteq [n]$, let $|S|$ denote the cardinality
  of $S$, $\spn(S)$ be the linear space spanned by $\{{\bf g}_i ~:~ i\in
  S\}$ over $\mathbb{F}_q$ and $\rank(S)$ be the dimension of
  $\spn(S)$.
\end{itemize}

\subsection{Locally Repairable Codes}

Let us recall some necessary definitions concerning locally repairable
codes. Assume throughout that $\mathcal{C}$ be an $[n,k,d]_q$ linear
code with generator matrix $G=({\bf g}_1,{\bf g}_2,\cdots,{\bf
  g}_{n})$.
\begin{definition}[\cite{HCL,PKLK}]\label{def_locality}
The $i$th code symbol of an $[n, k, d]_q$ linear code $\mathcal{C}$, is said
to have $(r, \delta)$-locality  if
there exists a subset $S_i\subseteq
[n]$ (a \emph{repair set}) such that
\begin{itemize}
  \item $i\in S_i$ and $|S_i|\leq r+\delta-1$; and
  \item The minimum Hamming distance of the punctured code
    $\mathcal{C}|_{S_i}$, obtained by deleting the code symbols $c_j$
    for all $j \in [n]\setminus S_i$, is at least $\delta$.
\end{itemize}
Furthermore, an $[n,k,d]_q$ linear code $\mathcal{C}$ is said to have
information $(r,\delta)$-locality (denoted as $(r,\delta)_i$-locality)
if there exists a $k$-subset $I\subseteq [n]$ with $\rank(I)=k$ such that
for each $i\in I$, the $i$th code symbol has $(r, \delta)$-locality
and all symbol $(r,\delta)$-locality (denoted as
$(r,\delta)_a$-locality) if all the $n$ code symbols have
$(r,\delta)$-locality.
\end{definition}

In \cite{PKLK} (also, \cite{GHSY} for $\delta=2$), an upper bound on
the minimum Hamming distance of linear codes with
$(r,\delta)_i$-locality was derived as follows.
\begin{lemma}[\cite{PKLK}] \label{lemma_bound_i}
The minimum distance of an $[n,k,d]_q$ code $\mathcal{C}$ with
$(r,\delta)_i$-locality is upper bounded by
\begin{equation}\label{eqn_bound_all_lcoal}
d\leq n-k+1-\left(\left\lceil\frac{k}{r}\right\rceil-1\right)(\delta-1).
\end{equation}
\end{lemma}
\begin{definition}
A linear code with $(r,\delta)_i$-locality is said to be an
\emph{optimal locally repairable code} if its minimum Hamming
distance meets the Singleton-type bound of Lemma~\ref{lemma_bound_i}
with equality.
\end{definition}

According to \eqref{eqn_bound_all_lcoal}, even for an optimal
$[n,k,d]_q$ linear code with $(r,\delta)_i$-locality (or
$(r,\delta)_a$-locality), $d<n-k+1$ under the nontrivial case $k>r$.
Thus, for a linear code with $(r,\delta)_i$-locality, it is natural to
ask if it is possible for an erasure pattern $E\subset [n]$ with size
$d\leq |E|\leq n-k$ to be recoverable \cite{GHSY}. Although this
problem is still open in general, for the following two special
settings received special attention in some previous works:

{\bf Setting I:} (e.g., \cite{BHH,PB}) For a linear code with $(r,\delta)_a$-locality, let
$(r+\delta-1)|n$ and $|\{S_i~:~i\in [n]\}|=\frac{n}{r+\delta-1}$,
i.e., all the $n$ symbols are divided into $\frac{n}{r+\delta-1}$
repair sets. Let $s=\frac{n}{r+\delta-1}r-k$ and assume the elements
of $S_i$ are denoted by $\{s_{i,1},s_{i,2},\dots,s_{i,r+\delta-1}\}$. An erasure
pattern $E$ is required to be recoverable if
there exists a $(\delta-1)$-subset of $[r+\delta-1]$, $\{j_1,j_2,\cdots,j_{\delta-1}\}$, and there exists a set $E^*\subseteq E\subseteq [n]$,
$|E^*|\leq \frac{n}{r+\delta-1}r-k$ and
\[(E\backslash E^*)\cap S_i\subseteq \{s_{i,j_1},s_{i,j_2},\dots,s_{i,j_{\delta-1}}\}\text{ for each }i\in [n].\]

{\bf Setting II:} (e.g., \cite{GHJY}) For a linear code with
$(r,\delta)_a$-locality, let $(r+\delta-1)|n$ and $|\{S_i~:~i\in
[n]\}|=\frac{n}{r+\delta-1}$, i.e., all the $n$ symbols are divided
into $\frac{n}{r+\delta-1}$ repair sets. Let
$s=\frac{n}{r+\delta-1}r-k$. An erasure pattern $E$ is required to be
recoverable if there exists a set $E^*\subseteq
E\subseteq [n]$, $|E^*|\leq s$ and
\[|(E\backslash E^*)\cap S_i|\leq\delta-1\text{ for each }1\leq i\leq
\frac{n}{r+\delta-1}.\]

\begin{definition}
An $[n,k,d]_q$ linear code that satisfies the conditions of Setting I
is said to be a sector-disk code ($(\delta-1,s)$-SD).
\end{definition}

As an intuition, we make the following analogies between a distributed
storage system and Setting I. In this analogy, we have a total of
$r+\delta-1$ disks, each containing $\frac{n}{r+\delta-1}$ sectors,
with a total number of sectors in the system which is $n$. The $i$th
stripe, i.e., the set containing the $i$th sector from each disk, is
an $(r,\delta)$-repair set, for each $i$. Finally, an SD code is
capable of correcting $\delta-1$ whole disk erasures, as well as an
extra $s$ erased sectors.

\begin{definition}
An $[n,k,d]_q$ linear code that satisfies the conditions of Setting II
is said to be a maximally recoverable code (MR code).
\end{definition}

MR codes are also known as partial MDS (PMDS) codes
\cite{BHH,BPSY,CK}. It is easy to check that Setting I is a special
case of Setting II, thus, MR codes are also SD codes, but not vice
versa.  For explicit constructions, the reader may refer to
\cite{BPSY,PB} for SD codes, and \cite{BHH,BPSY,CK,GYBS,MK} for MR
codes. Finally, another example of a family of codes that may recover
special erasure patterns beyond the minimum Hamming distance is STAIR
codes \cite{LL}.

\subsection{Packings and Steiner Systems}

We now turn to describe some definitions and known facts concerning
the combinatorial objects of packings and Steiner systems.

\begin{definition}[\cite{CD}, VI. 40]\label{def_packing}
Let $n\geq 2$ and $t,\tau$ be positive integers.  A
$\tau$-$(n,t,1)$-\emph{packing} is a pair $(X,\cB)$, where $X$ is a
set of $n$ elements (called points) and $\cB\subseteq 2^X$ is a
collection of $t$-subsets of $X$ (called blocks), such that each
$\tau$-subset of $X$ is contained in at most one block of $\cB$. If
$\tau=2$, it is also denoted as an $(n,t,1)$-packing. The packing is
said to be \emph{regular} if each element of $X$ appears in exactly
$w$ blocks, denoted as a $w$-regular $\tau$-$(n,t,1)$-packing.
\end{definition}



\begin{definition}[\cite{CD}, II. 5]\label{def_Packing}
Let $n\geq 2$ and $t,\tau$ be positive integers. A
$(\tau,t,n)$-\emph{Steiner system} is a pair $(X,\cB)$, where $X$ is a
set of $n$ elements (called points) and $\cB\subseteq 2^X$ is a
collection of $t$-subsets of $X$ (called blocks), such that each
$\tau$-subset of $X$ is contained in exactly one block of $\cB$.
\end{definition}

\begin{lemma}[\cite{CD}, II. 5]
A $(\tau,t,n)$-Steiner system is a $\frac{\binom{n-1}{\tau-1}}{\binom{t-1}{\tau-1}}$-regular $\tau$-$(n,t,1)$-packing.
\end{lemma}

\begin{remark}
Given positive integers $\tau$, $t$ and $n$, the natural necessary
conditions for the existence of a $(\tau,t,n)$-Steiner system are that
$\binom{t-i}{\tau-i}|\binom{n-i}{\tau-i}$ for all $0\leq i\leq
\tau-1$. It was shown in \cite{K} that these conditions are also
sufficient except perhaps for finitely many cases.
\end{remark}

\section{Constructions of Locally Repairable Codes}\label{sec-LRC}

In this section, we introduce a general construction of locally
repairable codes with information locality. Let $k=r\ell+v$ with $0<
v\leq r$ and $n=k+(\ell+1)(\delta-1)+h$ with $h\geq 0$, where all
parameters are integers.

\begin{construction}\label{cons_poly}
  Let the $k$ information symbols be partitioned into $\ell+1$ sets, say,
  \begin{align*}
    I^{(i)}&=\{I_{i,1},I_{i,2},\dots,I_{i,r}\}, \quad \text{for $i\in[\ell]$,}\\
    I^{(\ell+1)}&=\{I_{\ell+1,1},I_{\ell+1,2},\dots,I_{\ell+1,v}\}.
  \end{align*}
Let $S$ be an $h$-subset of $\F_q$ and denote $A\triangleq
\F_q\backslash S$. Let $\cA=\{A_i~:~1\leq i\leq \ell+1\}$ be a family
of subsets of $A$ with $|A_i|=r+\delta-1$ and
$|A_{\ell+1}|=v+\delta-1$.  Define
\begin{equation*}
g_i(x)=\prod_{\theta\in A_i}(x-\theta)\text{ for }1\leq i\leq \ell+1
\end{equation*}
and
\begin{equation*}
\Delta(x)=\prod_{1\leq i\leq \ell+1}g_i(x).
\end{equation*}

A linear code with length $n$ can be generated by defining a linear
map from the information ${\bm I}=(I_{1,1},\dots, I_{\ell,v})\in
\F^k_{q}$ to a codeword ${\bm C}({\bm
  I})=(c_{1,1},\dots,c_{\ell,r+\delta-1},c_{\ell+1,1},\dots,c_{\ell+1,v+\delta-1},c_{\ell+2,1},\dots,c_{\ell+2,h})\in\F^n_q$,
thus the $[n,k]_q$ linear code is $\cC=\{{\bm C}({\bm I}) ~:~ {\bm
  I}\in \F_q^{k}\}$. This mapping is performed by the following two
steps:

\paragraph{Step 1}
For $1\leq j\leq \ell+1$, by polynomial interpolation, there exists a
unique $f_j(x)\in\F_q [x]$ with $\deg(f_j)< |A_j|-\delta+1$ such that
$f_j(\theta_{j,t})=I_{j,t}$ for $1\leq t\leq |A_j|-\delta+1$, where
$A_j=\{\theta_{j,t}~:~ 1\leq t\leq |A_j|\}$.  For $1\leq j\leq \ell+1$ and
$1\leq t\leq |A_j|$, set $c_{j,t}=f_j(\theta_{j,t})$.

\paragraph{Step 2}
Let
\begin{equation*}
f_I(x)=\Delta(x)\sum_{1\leq i\leq \ell+1}\frac{f_i(x)}{g_i(x)}.
\end{equation*} Set $c_{\ell+2,i}=f_{I}(s_i)$ for $1\leq i\leq h$, where
$S=\{s_i~:~1\leq i\leq h\}$.
\end{construction}

\begin{lemma}\label{lemma_LRC_poly}
The code $\cC$ generated by Construction~\ref{cons_poly} is an
$[n,k]_q$ linear code with $(r,\delta)_i$-locality.
\end{lemma}
\begin{IEEEproof}
It is easy to verify that $\cC$ is an $[n,k]_q$ linear code. By
Construction~\ref{cons_poly}, Step 1, for any $C\in \cC$ and $1\leq
i\leq \ell+1$, $(c_{i,1},c_{i,1}, \dots,c_{i,|A_i|})$ is the
evaluation of a polynomial with degree at most $|A_i|-\delta$, which
means any $|A_i|-\delta+1\leq r$ components are capable of recovering
the remaining components. Thus, the code $\cC$ has
$(r,\delta)_i$-locality.
\end{IEEEproof}

For ease of presentation, we use the evaluation points (instead of the
indices of code symbols) to denote erasure patterns. Additionally, we
shall group the erased positions by the index of the repair set they
hit. Namely, we shall use $\cE=\mathset{E_1,\dots,E_{\ell+2}}$ to
denote an erasure pattern, where $E_j\subseteq A_j$ is the set of
erasure points in $A_j$, $1\leq j\leq \ell+1$, and
$E_{\ell+2}\subseteq S$ is the set of erasure points in $S$.

\begin{theorem}\label{theorem_e_pattern}
Let $\cC$ be the linear code generated by Construction~\ref{cons_poly}.  Assume $\cE=\{E_{i}~:~1\leq i\leq \ell+2\}$ is an
erasure pattern, with $E_{i}\subseteq A_{i}$ for $1\leq t\leq \ell+1$
and $E_{{\ell+2}}\subseteq S$. For $1\leq i\leq \ell+1$, assume that, in
$\cE$, there exist $w\leq \ell+1$ sets with $|E_{i_t}|\geq \delta$ for
$1\leq t\leq w$ and $1\leq i_t\leq \ell+1$.  If the erasure pattern
$\cE$ satisfies
\begin{equation}\label{eqn_e_num}
\abs{\bigcup_{1\leq t\leq w}E_{i_t}}+\abs{E_{\ell+2}}\leq h+\delta-1,
\end{equation}
and for any $1\leq j\leq w$
\begin{equation}\label{eqn_local_repair}
\abs{A_{i_j}\cap\parenv{\bigcup_{ j\ne t\in[w]}A_{i_t}}}\leq \delta-1,
\end{equation}
then the erasure pattern $\cE$ can be recovered.
\end{theorem}

\begin{remark}
Before proving Theorem~\ref{theorem_e_pattern}, we want to highlight
that the size $\abs{(\bigcup_{1\leq t\leq w}E_{i_t})\cup{E_{\ell+2}}}$
dictates whether an erasure pattern is recoverable, and not the number
of erased coordinates, i.e., $\sum_{1\leq t\leq
  w}|E_{i_t}|+|E_{\ell+2}|$. This is to say, if there are erasures that
share the same evaluation point (even in different coordinates), then
those erasures as a whole will only increase the discriminant value by
one. In such a case we may recover more than $h+\delta-1$ erasures
that are guaranteed to be recoverable by the value of the
Singleton-type bound, i.e., $h+\delta$.
\end{remark}

\begin{IEEEproof}
Since the linear code generated by Construction~\ref{cons_poly} has
$(r,\delta)_i$-locality, the locality is capable of recovering all the
erasures for the case $E_{i}\in \cE$ with $|E_{i}|\leq \delta-1$ and
$1\leq i\leq \ell+1$ independently.  Thus, in this proof we only need
to consider the case $E_{i}\in \cE$ with $|E_{i}|\geq \delta$ and
$1\leq i\leq \ell+1$, i.e., $E_{i_t}$ for $1\leq t\leq w$.

Let
$$\Phi(x)\triangleq\gcd\parenv{\frac{\Delta(x)}{g_{i_1}(x)},\frac{\Delta(x)}{g_{i_2}(x)},\cdots,\frac{\Delta(x)}{g_{i_w}(x)}}
=\frac{\Delta(x)}{\lcm\parenv{g_{i_1}(x),g_{i_2}(x),\ldots,g_{i_w}(x)}}
=\frac{\Delta(x)}{\prod_{\theta\in A}(x-\theta)},$$
where $A\triangleq \bigcup_{1\leq t\leq w} A_{i_t}$.
Considering $f_I(x)$ in Construction~\ref{cons_poly},
it can be rewritten as
\begin{equation*}
f_I(x)=\Phi(x)\sum_{1\leq t\leq w}f^*_{i_t}(x)+g(x)=\Phi(x)f_E(x)+g(x),
\end{equation*}
where
\begin{equation*}
g(x)=\Delta(x)\sum_{j\in [n]\setminus\{i_t~:~1\leq t\leq w\}}\frac{f_j(x)}{g_j(x)}
\end{equation*} is a known polynomial
determined by the known code symbols by Construction~\ref{cons_poly},
\begin{equation}\label{eqn_f^*_i}
f^*_{i_t}(x)=\frac{\Delta(x)f_{i_t}(x)}{\Phi(x)g_{i_t}(x)}=\frac{f_{i_t}(x)\prod_{\theta\in A}(x-\theta)}{g_{i_t}(x)},
\end{equation} and
\begin{equation}\label{eqn_deg_fi*}
\deg(f^*_{i_t}(x))=|A|-|A_{i_t}|+\deg(f_{i_t}(x))
\end{equation}
for $1\leq t\leq w$.
Recall that \eqref{eqn_local_repair} means that for $1\leq t\leq w$
there exists a set
$A^*_{i_t}\triangleq A_{i_t}\setminus \parenv{\bigcup_{j\ne t,1\leq j\leq w}A_{i_j}}$
 such that $|A^*_{i_t}|\geq r$.
 Let \begin{equation}\label{eqn_def_e_ij}
 e_{i_t,j}\triangleq\left.\frac{\prod_{\theta\in A}(x-\theta)}{g_{i_t}(x)}\right|_{x=\theta_{i_t,j}}
 \end{equation}
for $1\leq t\leq w$ and $\theta_{i_t,j}\in A_{i_t}$. Then, for
$\theta_{i_t,j}\in A^*_{i_t}$, we have
$$e_{i_t,j}c_{i_t,j}=e_{i_t,j}f_{i_t}(\theta_{i_t,j})=f^*_{i_t}(\theta_{i_t,j})=
f_E(\theta_{i_t,j}),$$
where the last equality holds by the fact that $f^*_{i_\tau}(\theta_{i_t,j})=0$ for $1\leq \tau\leq w$,
$\tau\ne t$, and $\theta_{i_\tau,j}\in A^*_{i_t}$.
Let $A^*=\bigcup_{1\leq t\leq w} A^*_{i_t}$.
Note that for $\theta\in A\setminus A^*$ and $1\leq t\leq w$,
$f^*_{i_t}(\theta)=0$ if $\theta\not \in A_{i_t}$.
For $\theta \in A\setminus A^*$, by \eqref{eqn_f^*_i} we have
\begin{equation*}
\sum_{\substack{\theta_{i_t,j}=\theta\in A_{i_t},\\
1\leq t\leq w}}e_{i_t,j}c_{i_t,j}=\sum_{\substack{\theta_{i_t,j}=\theta\in A_{i_t},\\
1\leq t\leq w}}e_{i_t,j}f_{i_t}(\theta_{i_t,j})=\sum_{\substack{\theta_{i_t,j}=\theta\in A_i,\\
1\leq t\leq w}}f^*_i(\theta_{i_t,j})=f_E(\theta),
\end{equation*}
where $e_{i_t,j}$ is defined by \eqref{eqn_def_e_ij}.  The last
equation implies that if we know all the code symbols in $A_{i_t}$ for
$1\leq t\leq w$ corresponding to the same element $\theta\in
A\setminus A^*$ then we know the value of $f_E(\theta)$.  In other
words, we know all the values $f_E(\theta)$ for $\theta \in
(A\setminus A^*) \setminus \parenv{\bigcup_{1\leq t\leq w}E_{i_t}}$.

Furthermore, for $\theta=\theta_{\ell+2,t}\in S\setminus E_{\ell+2}$,
we have $c_{\ell+2,t}=f_I(\theta)=\Phi(\theta)f_E(\theta)+g(\theta)$,
i.e., $f_E(\theta)=\frac{c_{\ell+2,t}-g(\theta)}{\Phi(\theta)}$, where
$g(x)$ can be regarded as a known polynomial.  Thus, under the erasure
pattern $\cE$, we know
\begin{equation*}
\begin{split}
&\abs{A^*\setminus\parenv{\bigcup_{1\leq t\leq w}E_{i_t}}}+\abs{(A\setminus A^*)\setminus \parenv{\bigcup_{1\leq t\leq w}E_{i_t}}}+|S|-|E_{\ell+2}|\\
=& \abs{A^*\setminus\parenv{\bigcup_{1\leq t\leq w}E_{i_t}}}+\abs{(A\setminus A^*)\setminus \parenv{\bigcup_{1\leq t\leq w}E_{i_t}}}+h-|E_{\ell+2}|\\
\geq& \abs{A^*\setminus\parenv{\bigcup_{1\leq t\leq w}E_{i_t}}}+\abs{(A\setminus A^*)\setminus \parenv{\bigcup_{1\leq t\leq w}E_{i_t}}}
+\abs{\parenv{\bigcup_{1\leq t\leq w}E_{i_t}}}-\delta+1\\
=& |A|-\delta+1\\
>& |A|-|A_i|+\deg(f_{i_t}(x))=\deg(f_{i_t}^*(x))\\
\end{split}
\end{equation*}
evaluation points and the corresponding value for $f_E(x)$, where the
two inequalities hold by \eqref{eqn_e_num} and \eqref{eqn_deg_fi*},
respectively. This is to say that we can recover $f_E(x)$ since
$\deg(f_E(x))\leq \max_{1\leq t\leq w}\deg(f^*_{i_t}(x))$.  Now the
fact that
\[\abs{A_{i_t}\cap\parenv{\bigcup_{1\leq j\leq w, j\ne
      t}A_{i_j}}}\leq \delta-1\]
implies that we can also recover
$f_{i_t}(x)$ for $1\leq t\leq w$ and all the code symbols in $A_{i_t}$
for $1\leq t\leq w$. Finally, by $f_I(x)$ we can recover all the code
symbols in $E_{\ell+2}$.
\end{IEEEproof}

\begin{corollary}\label{cor_LRC_poly}
If the set system $\cA$ of Construction~\ref{cons_poly} satisfies that
for any $\mu$-subset $D$ of $[\ell+1]$
\begin{equation}\label{eqn_cond_A}
\abs{A_i\cap \parenv{\bigcup_{j\ne i,j\in D} A_j}}\leq \delta-1 \quad \text{for }i\in D,
\end{equation}
then the code $\cC$ generated by Construction~\ref{cons_poly} is an
$[n,k,d]_q$ linear code with $(r,\delta)_i$ locality and $d\geq
\min\{(\mu+1)\delta,h+\delta\}$.  Furthermore, if $h+\delta\leq
(\mu+1)\delta$, then the code $\cC$ is optimal with respect to the bound
in Lemma~\ref{lemma_bound_i}.
\end{corollary}
\begin{IEEEproof}
By Lemma~\ref{lemma_LRC_poly}, we only need to prove $d\geq
\min\{(\mu+1)\delta,h+\delta\}$. To bound the minimum Hamming distance
of $\cC$, we consider the following two cases:

For the case $h+\delta\geq (\mu+1)\delta$, we prove that the code
$\cC$ is capable of recovering any erasure pattern
$\cE=\{E_{i}~:~1\leq i\leq \ell+2\}$ with $\sum_{1\leq i\leq
  \ell+1}|E_i|+|E_{\ell+2}|\leq (\mu+1)\delta-1$, where $E_i\subseteq
A_i$ for $1\leq i\leq \ell+1$ and $E_{\ell+2}\subseteq S$.  Note that
the $(r,\delta)_i$-locality means that we only need to consider the
case that $|E_{i_t}|\geq \delta$ for $1\leq i_t\leq \ell+1$ and
$1\leq t\leq w$. Note that $\sum_{1\leq i\leq
  \ell+1}|E_i|+|E_{\ell+2}| \leq (\mu+1)\delta-1$ implies that $w\leq
\mu$. By \eqref{eqn_cond_A}, we may conclude that
$$\abs{A_{i_t}\cap\parenv{\bigcup_{ t\ne j\in[w]}A_{i_j}}}\leq \delta-1.$$ Now the fact that
$|E_{u+2}|+\sum_{1\leq t\leq w}|E_{i_t}|\leq \sum_{1\leq i\leq
  \ell+1}|E_{i}|\leq \mu\delta\leq h+\delta-1$ implies that $\cE$ is
recoverable by Theorem~\ref{theorem_e_pattern}.

For the case $(\mu+1)\delta> h+\delta$, similarly, we are going to
prove that the code $\cC$ is capable of recovering any erasure pattern
$\cE=\{E_{i}~:~1\leq i\leq \ell+2\}$ with $\sum_{1\leq i\leq
  \ell+2}|E_i|\leq h+\delta-1$, where $E_i\subseteq A_i$ for $1\leq
i\leq \ell+1$ and $E_{\ell+2}\subseteq S$.  Similarly, we conclude
that there are at most $w\leq \mu$ sets
$E_{i_1},E_{i_2},\cdots,E_{i_w}$ with $|E_{i_l}|\geq \delta$ and
$1\leq i_t\leq \ell+1$ for $1\leq t\leq w$.  Again by
Theorem~\ref{theorem_e_pattern}, \eqref{eqn_cond_A}, and the fact that
\[|E_{u+2}|+\sum_{1\leq t\leq w}|E_{i_t}|\leq \sum_{1\leq i\leq \ell+2}|E_{i}|\leq h+\delta-1,\]
we have that $\cE$ is recoverable.

Finally, the optimality of $\cC$ follows directly from
Lemma~\ref{lemma_bound_i} and the fact that $h+\delta\leq
(\mu+1)\delta$.
\end{IEEEproof}

\subsection{Optimal Locally Repairable Codes with $(r,\delta)_i$-Locality Based on Packings or Steiner Systems}

Based on Corollary~\ref{cor_LRC_poly}, to construct optimal locally
repairable codes we only need to find $\cA$ such that
\eqref{eqn_cond_A} holds.  In this section, we consider the case that
$\cA$ forms a combinatorial structure which satisfies the
condition given by \eqref{eqn_cond_A}. We first consider a condition
on the intersection of any pair of sets in $\cA$ rather than $\mu$
sets as in \eqref{eqn_cond_A}.

\begin{theorem}\label{theorem_optimal_code_cond_a}
  Assume the setting of Construction~\ref{cons_poly}. Let $\cA$ be a
  set system formed by subsets of $\F_q\setminus S$, where $S$ is an
  $h$-subset of $\F_q$.  If there exists a positive integer $a$ such
  that $|A_i\cap A_j|\leq a$ for all $i\neq j$, then the
  code $\cC$ generated by Construction~\ref{cons_poly} is an
  $[n,k,d\geq \min\{h+\delta,
    (\lceil\frac{\delta}{a}\rceil+1)\delta\}]_q$ linear code with
  $(r,\delta)_i$-locality. If additionally,
  $h\le\lceil\frac{\delta}{a}\rceil\delta$, then the code $\cC$
  generated by Construction~\ref{cons_poly} is an optimal
  $[n,k,d=h+\delta]_q$ linear code with $(r,\delta)_i$-locality.
\end{theorem}
\begin{proof}
Let $\mu=\lceil\frac{\delta}{a}\rceil$. Then for any $\mu$-subset, $\cR\subseteq\cA$, and
for any $A'\in\cR$, we have
\begin{equation*}
  \left|A'\cap \left(\bigcup_{A\in\cR\setminus\mathset{A'}}A\right)\right|\leq
  (\mu-1)a=\left(\left\lceil\frac{\delta}{a}\right\rceil-1\right)a
  \leq \delta-1,
\end{equation*}
since $|S_i\cap S_j|\leq a$. The first claim follows from
Corollary~\ref{cor_LRC_poly}.  Note that $\mu\delta\geq
\lceil\frac{\delta}{a}\rceil\delta\ge h$ means that $(\mu+1)\delta\geq
h+\delta$. Again by Corollary~\ref{cor_LRC_poly} we have the desired
result follows.
\end{proof}

Based on Theorem~\ref{theorem_optimal_code_cond_a}, we can use
combinatorial designs to generate optimal locally repairable codes via
Construction~\ref{cons_poly}. The following corollaries follow
directly from Theorem~\ref{theorem_optimal_code_cond_a}.

\begin{corollary}\label{corollary_optimal_code_packing}
  Let $S$ be an $h$-subset of $\F_q$.  If there exists a
  $(\tau+1)$-$(q-h,r+\delta-1,1)$-packing $(\F_q\setminus S,\cB)$ and
  $0 \le h \le \lceil\frac{\delta}{\tau}\rceil\delta$, then there
  exists an optimal $[n,k,d]_q$ linear code with
  $(r,\delta)_i$-locality, where $n=|\cB|(r+\delta-1)+h-r+v$,
  $k=(|\cB|-1)r+v$, and $d=h+\delta$.
\end{corollary}


\begin{corollary}
  If there exists a $(\tau+1,r+\delta-1,q-h)$-Steiner system and $0\le
  h\leq \lceil\frac{\delta}{\tau}\rceil\delta$, then there exists an
  optimal $[n,k,d]_q$ linear code with $(r,\delta)_i$-locality, where
\begin{align*}
n&=\frac{\binom{q-h}{\tau+1}(r+\delta-1)}{\binom{r+\delta-1}{\tau+1}}+h+v-r,\\
k&=\parenv{\frac{\binom{q-h}{\tau+1}}{\binom{r+\delta-1}{\tau+1}}-1}r+v,
\end{align*}
and $d=h+\delta$.
\end{corollary}

\subsection{Optimal Locally Repairable Codes with Order-Optimal Length: $(r,\delta)_i$-Locality}

Finding the maximal length of optimal locally repairable codes with
$(r,\delta)_a$-locality was the subject of \cite{GXY,CMST_SL}, for the
cases of $\delta=2$ and $\delta\geq 2$, respectively. Both
constructions and bounds are proposed there. It is therefore natural
to further ask how long can optimal locally repairable codes with
$(r,\delta)_i$-locality be. This question is also important to us in
order to analyze the performance of Construction~\ref{cons_poly}.

\begin{theorem}\label{theorem_bound_delta>2}
Let $n=k+\ell (\delta-1)+h$, $\delta\geq 2$, $k=\ell r$. Assume there
exists an optimal $[n,k,d]_q$ linear code $\cC$ with
$(r,\delta)_i$-locality. For any given integer $0\leq a\leq h$
define $T(a)=\floorenv{(d-a-1)/\delta}$.  If
$T(a)\geq 2$, then
\begin{align*}
n&\leq
\begin{cases}
\frac{r+\delta-1}{r}\parenv{\frac{T(a)-1}{2(q-1)}q^{\frac{2(h-a-1)}{T(a)-1}}+a+1}-\frac{h(\delta-1)}{r}, &\text{ if } T(a) \text{ is odd},  \\
\frac{r+\delta-1}{r}\parenv{\frac{T(a)}{2(q-1)}q^{\frac{2(h-a)}{T(a)}}+a}-\frac{h(\delta-1)}{r}, &\text{ if } T(a) \text{ is even},\\
\end{cases}
\end{align*}
where $h$ can be rewritten as $h=d-\delta$.
\end{theorem}

The technical proof and its supporting lemmas are included in
Appendix~\ref{app:proof1}.

Throughout the paper we shall look at the asymptotics of families of
codes with locality. In the terminology of
Theorem~\ref{theorem_bound_delta>2} we assume $r,\delta,h,d$ (and
therefore $a$) are all constants. If the codes we study are optimal
(with respect to the bound of Lemma~\ref{lemma_bound_i}), then $k$ may
be derived from $n$. Thus, we are left with the asymptotics of $n$ as
a function of the field size $q$. Therefore, we shall say the family
of codes is \emph{order optimal} if, up to a constant factor, it
attains the bound of Theorem~\ref{theorem_bound_delta>2}, namely,
\[
n=
\begin{cases}
\Theta\parenv{q^{\frac{2(h-a-1)}{T(a)-1}-1}} &\text{if $T(a)$ is odd},  \\
\Theta\parenv{q^{\frac{2(h-a)}{T(a)}-1}} &\text{if $T(a)$ is even}.
\end{cases}
\]


Now, based on Theorem~\ref{theorem_bound_delta>2}, we can analyze the
performance of Construction~\ref{cons_poly}. The number of blocks of a
packing is upper bounded by the following Johnson bound \cite{J}:

\begin{lemma}[\cite{J}]\label{lemma_johnson}
  The maximum possible number of blocks of a
  $(\tau+1)$-$(n_1,r+\delta-1,1)$-packing $(X,\cB)$ is bounded by
\begin{equation*}
\begin{split}
|\cB|
\leq\left\lfloor \frac{n_1}{r+\delta-1} \left\lfloor\frac{n_1-1}{r+\delta-2} \left\lfloor\frac{n_1-2}{r+\delta-3}
\dots\left\lfloor\frac{n_1-\tau}{r+\delta-1-\tau} \right\rfloor\dots\right\rfloor\right\rfloor\right\rfloor.
\end{split}
\end{equation*}
\end{lemma}

Thus, the number of blocks for a
$(\tau+1)$-$(n_1,r+\delta-1,1)$-packing can be as large as
$O(n_1^{\tau+1})$, when $\tau$, $r$, and $\delta$ are regarded as constants.

\begin{corollary}\label{corollary_code_via_packings}
Let $n_1=q-h$. If there exists a
$(\tau+1)$-$(n_1,r+\delta-1,1)$-packing with blocks $\cB$,
$|\cB|=\Omega(n_1^{\tau+1})$,
and $0 \le h \le \lceil\frac{\delta}{\tau}\rceil\delta$, then there
exists an optimal $[n,k,d]_q$ linear code with
$(r,\delta)_i$-locality, where
$n=|\cB|(r+\delta-1)+h+v-r=\Omega(q^{\tau+1})$, $k=(|\cB|-1)r+v$ and
$d=h+\delta$. Furthermore, if $h\geq \delta+1$, $v=r$, and $\tau=\delta-1$ the
code based on the $(\tau+1)$-$(n_1,r+\delta-1,1)$-packing has
order-optimal length, where $r$, $h$, and $\delta$ are
regarded as constants.
\end{corollary}
\begin{proof}
By Corollary~\ref{corollary_optimal_code_packing}, we have
$n=|\cB|(r+\delta-1)+h+v-r=\Omega(q^{\tau+1})$ for the code generated by
Construction~\ref{cons_poly}.
In Theorem~\ref{theorem_bound_delta>2}, setting $a=h-\delta-1$,
we have $T(a)=\lfloor\frac{d-1-a}{\delta}\rfloor=\lfloor\frac{h+\delta-a-1}{\delta}\rfloor=2$.
Therefore, for the case $v=r$, by Theorem~\ref{theorem_bound_delta>2} again
$$n\leq
\frac{r+\delta-1}{r}\parenv{\frac{T(a)}{2(q-1)}q^{\frac{2(h-a)}{T(a)}}+a}-\frac{h(\delta-1)}{r}
=\frac{r+\delta-1}{r}\parenv{\frac{1}{q-1}q^{\delta+1}+a}-\frac{h(\delta-1)}{r}=O(q^\delta).$$
Thus, for the case $\tau=\delta-1$ and $v=r$, the code $\cC$ has length
$n=\Omega(q^{\tau+1})=\Omega(q^{\delta})$, which is order optimal with
respect to the bound in Theorem~\ref{theorem_bound_delta>2}, when $h$,
$r$, and $\delta$ are regarded as constants.
\end{proof}

\begin{corollary}\label{coro_steiner}
Let $n_1=q-h$. If there exists a $(\tau+1,r+\delta-1,n_1)$-Steiner
system and $0\le h\leq \lceil\frac{\delta}{\tau}\rceil\delta$,
then there exists an
optimal $[n,k,d]_q$ linear code with all symbol $(r,\delta)$-locality,
where
\begin{align*}
n&=\frac{\binom{n_1}{\tau+1}(r+\delta-1)}{\binom{r+\delta-1}{\tau+1}}+h,\\
k&=\frac{r\binom{n_1}{\tau+1}}{\binom{r+\delta-1}{\tau+1}},
\end{align*}
and $d=h+\delta$.  In particular, for the case $h\geq \delta+1$ and
$\tau=\delta-1$, the code based on the
$(\delta,r+\delta-1,q-h)$-Steiner system has order-optimal length,
where $h$, $r$, and $\delta$ are regarded as constants.
\end{corollary}
\begin{proof}
The first part of the corollary follows directly from
Corollary~\ref{corollary_optimal_code_packing} and
Definition~\ref{def_Packing}. For the second part, the fact $h\geq
\delta+1$ means that we can set $a=h-\delta-1$ and
$T(a)=\lfloor\frac{d-1-a}{\delta}\rfloor=\lfloor\frac{h+\delta-a-1}{\delta}\rfloor=2$ in
Theorem~\ref{theorem_bound_delta>2}, which also means the code $\cC$
has length
$$n\leq
\frac{r+\delta-1}{r}\parenv{\frac{T(a)}{2(q-1)}q^{\frac{2(h-a)}{T(a)}}+a}-\frac{h(\delta-1)}{r}
=\frac{r+\delta-1}{r}\parenv{\frac{1}{q-1}q^{\delta+1}+a}-\frac{h(\delta-1)}{r}=O(q^\delta).$$
Now the conclusion comes from the fact
that the upper bound is $O(q^\delta)$ and the constructed code has length
$n=\frac{\binom{n_1}{\tau+1}(r+\delta-1)}{\binom{r+\delta-1}{\tau+1}}+h=\Omega(q^\delta)$,
where we assume $h$, $r$, and $\delta$ are constants.
\end{proof}

\begin{remark}
For the existence of packings in general the reader may refer to \cite{R} and
the survey in \cite[VI.40]{CD}.
\end{remark}

\begin{remark}\label{rem_steiner}
Given positive integers $\tau$, $r$ and $\delta>2$, the natural
    necessary conditions for the existence of a
$(\tau+1,r+\delta-1,t-r+v)$-Steiner system are that
$\binom{r+\delta-1-i}{\tau+1-i}|\binom{t-r+v-i}{\tau+1-i}$ for
all $0\leq i\leq \tau$. It was shown in \cite{K} that these
  conditions are also sufficient except perhaps for finitely many
  cases. While $q$ might not be a prime power, any prime power $\overline{q}\geq
  q$ will suffice for our needs. It is known, for example, that there
  is always a prime in the interval $[q,q+q^{21/40}]$ (see
  \cite{BHP}).  Thus, by Corollary~\ref{coro_steiner}, for all large enough $t$, there exists an
optimal $[n,k,d]_{q}$ locally repairable code, with
$(r,\delta)_i$-locality, where $q$ is a prime power with $t\leq q\leq t+t^{21/40}$
and \begin{align*}
   n&=(r+\delta-1)\cdot\frac{\binom{t-r+v}{\tau+1}}{\binom{r+\delta-1}{\tau+1}}+h=\Omega(t^{\tau+1})=\Omega(q^{\tau+1}),\\
    k&= \frac{r\binom{t-r+v}{\tau+1}}{\binom{r+\delta-1}{\tau+1}},\\
    d&=h+\delta.
  \end{align*}
\end{remark}

\begin{remark}
One well known construction for optimal locally repairable codes with
$(r,\delta)_i$-locality is that of pyramid codes. The pyramid code is
based on an MDS code whose length is upper bounded by $q+d-2$ (and by
the MDS conjecture this may be reduced to $q+1$ for $q$ odd
\cite{B}). Thus, the length of pyramid code is upper bounded by
$q+d-1-\delta+\frac{k}{r}\delta\leq q+d-1-\delta+\frac{q-1}{r}\delta$
(we note that $q+2-\delta+\frac{k}{r}\delta\leq
q+2-\delta+\frac{q-d+2}{r}\delta$ according to MDS conjecture for the
case of $q$ odd), where $d\geq \delta$. According to our construction
and bound (in Theorem~\ref{theorem_bound_delta>2}), it follows that
the pyramid code is sub-optimal in terms of asymptotic length, since we
construct locally repairable codes with $(r,\delta)_i$-locality and
length $n=\Omega(q^{\delta})$.
\end{remark}

\begin{example}\label{example_LRC}
Set $n=24$, $k=14$, $\delta=2$, $r=2$, and $h=3$. Let $\cA=\{A_i~:~
A_i\triangleq \{3,6,5\}+i\subseteq \Z_7, i\in \Z_7\}$.  According to
Construction \ref{cons_poly}, we can construct a linear code $\cC$
with $(2,2)_i$-locality over $\F_{11}$, whose parity check matrix can
be given as:
\begin{equation*}
H=\parenv{\begin{array}{cccccccccccccccccccccccccccccccc}
 7  & 5  & 0  & 0  & 0  & 0  & 0  & 0  & 0  & 0  & 0  & 0  & 0  & 0  & 10  & 0  & 0  & 0  & 0  & 0  & 0  & 0  & 0  & 0  \\
0  & 0  & 7  & 5  & 0  & 0  & 0  & 0  & 0  & 0  & 0  & 0  & 0  & 0  & 0  & 10  & 0  & 0  & 0  & 0  & 0  & 0  & 0  & 0  \\
0  & 0  & 0  & 0  & 7  & 5  & 0  & 0  & 0  & 0  & 0  & 0  & 0  & 0  & 0  & 0  & 10  & 0  & 0  & 0  & 0  & 0  & 0  & 0  \\
0  & 0  & 0  & 0  & 0  & 0  & 3  & 9  & 0  & 0  & 0  & 0  & 0  & 0  & 0  & 0  & 0  & 10  & 0  & 0  & 0  & 0  & 0  & 0  \\
 0  & 0  & 0  & 0  & 0  & 0  & 0  & 0  & 3  & 9  & 0  & 0  & 0  & 0  & 0  & 0  & 0  & 0  & 10  & 0  & 0  & 0  & 0  & 0  \\
0  & 0  & 0  & 0  & 0  & 0  & 0  & 0  & 0  & 0  & 9  & 3  & 0  & 0  & 0  & 0  & 0  & 0  & 0  & 10  & 0  & 0  & 0  & 0  \\
0  & 0  & 0  & 0  & 0  & 0  & 0  & 0  & 0  & 0  & 0  & 0  & 7  & 5  & 0  & 0  & 0  & 0  & 0  & 0  & 10  & 0  & 0  & 0  \\
 5  & 10  & 2  & 7  & 7  & 1  & 1  & 6  & 3  & 5  & 1  & 3  & 2  & 8  & 0  & 0  & 0  & 0  & 0  & 0  & 0  & 10  & 0  & 0  \\
 1  & 4  & 8  & 5  & 1  & 9  & 3  & 3  & 6  & 3  & 4  & 2  & 9  & 5  & 0  & 0  & 0  & 0  & 0  & 0  & 0  & 0  & 10  & 0  \\
2  & 6  & 10  & 7  & 3  & 6  & 7  & 3  & 8  & 8  & 7  & 2  & 4  & 2  & 0  & 0  & 0  & 0  & 0  & 0  & 0  & 0  & 0  & 10  \\
\end{array}
}.
\end{equation*}
Verified by a computer program, the minimum Hamming distance of $\cC$ is $5$.
Thus, in this setting, Construction \ref{cons_poly} generates a $[24,14,5]_{11}$
optimal linear code with $(2,2)_i$-locality, consistent with the
result in Theorem \ref{theorem_optimal_code_cond_a}. Note that,
to construct a code sharing the same parameters via the pyramid code,
we need an MDS code with parameters $[18,14,5]_q$. However, according
to the MDS conjecture this MDS code exists only under the condition
that $q\geq 17$. Without the help of MDS conjecture, based on the result
proposed in \cite{B}, we have $q\geq 16$ for this special setting.
\end{example}

\begin{remark}
For the case $\delta=2$ and $d=5$, optimal linear codes with all
symbol $(r,2)$-locality and order-optimal length $\Theta(q^2)$ have
been introduced in \cite{GXY,Jin,BCGLP}.  The constructions in
\cite{Jin,BCGLP} are given by parity-check matrices with $3$ or $4$
global parity checks, which means they only work for the cases
$d=5,6$. One can verify that our construction still works for more
general cases even if we restrict to the case $\delta=2$.
\end{remark}

\begin{remark}
For the case $\delta\geq 2$ and $d=2\delta+1$, optimal linear codes
with all symbol $(r,2)$-locality and order-optimal length
$\Theta(q^{\delta})$ have been introduced in \cite{CMST_SL}.  However,
the construction in \cite{CMST_SL} should satisfy the condition $h\leq
r+\delta-1$, which is not need for Construction \ref{cons_poly}.
\end{remark}

\section{Generalized Sector-Disk Codes}\label{sec-GSD}
By Theorem~\ref{theorem_e_pattern}, we may have extra benefits if
$\abs{\bigcup_{|E_{i}|\geq \delta, i\in[\ell+1]}E_i}<\sum_{|E_{i}|\geq \delta, i\in[\ell+1]}|E_i|$.  In this section, we are going to use this property
to construct array codes that can recover from special erasure
patterns beyond the minimum Hamming distance. The basic idea of those
construction is to let all the code symbols share the same evaluation
point in step 1 of Construction~\ref{cons_poly} in the same column of
an array code. Then for this array code, one erased column may only
increase the value $\abs{\bigcup_{|E_{i}|\geq \delta, i\in[\ell+1]}E_i}$ by one.  Hence, when we consider sector-disk-like
erasure patterns, we may get some extra benefit beyond the minimum
Hamming distance. We begin with some definitions.

\begin{definition}\label{def_GSD}
Let $\cC$ be an optimal $[n,k,d]_q$ linear code with
$(r,\delta)_i$-locality. Then the code $\cC$ is said to be an
$(s,\gamma)$-\emph{generalized sector-disk code} (GSD code) if the
codewords can be arranged into an array
\begin{equation*}
C=\parenv{\begin{matrix}
c_{1,1} & c_{1,2}&\cdots& c_{1,a}\\
c_{2,1} & c_{2,2}&\cdots& c_{2,a}\\
\vdots &\vdots&\ddots&\vdots\\
c_{b,1} & c_{b,2}&\cdots& c_{b,a}\\
\end{matrix}}
\end{equation*}
such that:
\begin{itemize}
  \item[(I)] {all the erasure patterns that contain any $s$ columns and
    additional $\gamma$ cells can be recovered; and}
  \item[(II)] {$sb+\gamma>d-1.$}
\end{itemize}
\end{definition}

\begin{remark}
If the code $\cC$ has $(r,\delta)_a$-locality, the repair sets are
exactly the rows, and then the $(d-\delta,\delta-1)$-GSD code is
exactly the $(d-\delta,\delta-1)$-SD code \cite{PB}. Compared with SD
codes, GSD codes relax the conditions in the following three aspects:
\begin{itemize}
  \item GSD codes only require $(r,\delta)_i$-locality, whereas SD
    codes require $(r,\delta)_a$-locality;
  \item A row in an array codeword of a GSD code is not necessary a repair set;
  \item The number of column erasures is not restricted to $\delta-1$
    as in SD codes.
\end{itemize}
\end{remark}

In the following construction, we use Construction~\ref{cons_poly} to
generate GSD codes.

\begin{construction}\label{cons_array}
Let $S$ be an $h$-subset of $\F_q$ and let $(X=\F_q\setminus
S,\cA=\{A_i~:~1\leq i\leq \ell+1\})$ be a $t$-regular
$(m,r+\delta-1,1)$-packing, where $A_i=\{\theta_{i,j}~:~1\leq j\leq
r+\delta-1\}$ for $1\leq i\leq \ell+1$.  Based on $\cA$ and $S$, we
can generate a locally repairable code $\cC$ according to
Construction~\ref{cons_poly}.  Define column vectors $V_\tau\in
\F^t_q$ for $\tau\in \F_q$ as
\begin{equation*}
V^\intercal_{\tau}=(c_{i_{\tau,1},j_{\tau,1}},c_{i_{\tau,2},j_{\tau,2}},\dots,c_{i_{\tau,t},j_{\tau,t}}),
\end{equation*}
where
\begin{equation*}
\theta_{i_{\tau,b},j_{\tau,b}}=\tau,\,\,\text{for}\,\,1\leq b\leq t.
\end{equation*}
Arrange the $h$ global parity symbols as the last
$\lceil\frac{h}{t}\rceil$ columns. If there are empty cells in the array, then we fill
them with $0$.
\end{construction}

\begin{theorem}\label{thm_array}
Let $\cC$ be the $t\times (m+\lceil\frac{h}{t}\rceil)$ array code
generated by Construction~\ref{cons_array}. Then each element of the
first $m$ columns has $(r,\delta)$-locality. If $h\leq \delta^2$, then
the code can recover from any $h+\delta-1$ erasures. Furthermore:
\begin{itemize}
  \item [(I)] The code $\cC$ can recover from any erasure pattern of
    $y\leq 2$ columns from the first $m$ columns and any other $h-y-1$
    erasures.
  \item [(II)] If $\binom{y}{2}\leq \delta$, then the code $\cC$ can
    recover from any erasure pattern of $y$ columns from the first $m$
    columns and any other $h-2-\binom{y}{2}$ erasures.
\item [(III)] The code $\cC$ can recover from any erasure pattern of
  $y<\frac{(\delta+1)\delta}{2}-1$ columns from the first $m$ columns
  and any other $\min\{\frac{(\delta+1)\delta}{2}-y-1,h+\delta-1-y\}$
  erasures.
\end{itemize}
\end{theorem}
\begin{IEEEproof}
By Lemma~\ref{lemma_LRC_poly} and Theorem~\ref{theorem_e_pattern}, we
only need to prove that the desired erasure patterns satisfy
\eqref{eqn_e_num} and \eqref{eqn_local_repair}. Since $\cA$ forms a
$t$-regular $(m,w,1)$-packing, for the condition given by
\eqref{eqn_local_repair} we consider a sufficient condition that is
the erasure pattern contains at most $\delta$ repair sets with each of
them containing more than $\delta$ erasures.

For case (I) and $y=2$, say the erased columns are marked by
$\theta_1$ and $\theta_2$.  We focus on the repair sets with more than
$\delta$ erasures. In those repair sets, there is at most one repair
set that contains $\theta_1$ and $\theta_2$, while the remaining
repair sets contain at most one of them. For this case, we need at
least $\delta-2+(\delta-1)(\delta-1)+\delta-2=\delta^2-3\geq h-3$
erasures before we achieve $\delta+1$ repair sets with each of them
containing more than $\delta$ erasures. Thus, the code $\cC$ can
recover from any erasure pattern of $y=2$ columns from the first $m$
columns and any other $h-3$ erasures.  The same analysis proves the
case $y=1$.

For the case (II), we assume that the erased columns are marked by
elements in $\Theta=\{\theta_1,\theta_2,\cdots,\theta_y\}$.  Note that
$\cA$ is an $(m,r+\delta-1,1)$-packing, which means that each
$2$-subset of $\Theta$ appears in at most one repair set. This is to
say, for any $\delta$ repair sets
$A_{j_1},A_{j_2},\cdots,A_{j_\delta}$ we have $$\sum_{1\leq i\leq
  \delta}|\Theta\cap A_{j_i}|\leq
2\binom{y}{2}+\delta-\binom{y}{2}=\binom{y}{2}+\delta,$$ which means
for any $E_{j_i}\subseteq A_{j_i}$
$$\sum_{1\leq i\leq \delta}|\Theta\cap E_{j_i}|\leq
2\binom{y}{2}+\delta-\binom{y}{2}=\binom{y}{2}+\delta.$$ Therefore,
for this case, we need at least
$\delta^2+\delta-2-\binom{y}{2}-\delta\geq h-2-\binom{y}{2}$ erasures
before we achieve $\delta+1$ repair sets each of which contains more
than $\delta$ erasures. In other words, the code $\cC$ can recover
from any erasure pattern of $y$ columns from the first $m$ columns and
any other $h-2-\binom{y}{2}$ erasures.

For the case (III), we assume that the erased columns are marked by
elements in $\Theta=\{\theta_1,\theta_2,\cdots,\theta_y\}$.  Note that
$\cA$ is an $(m,w,1)$-packing, which means that $A_i\cap A_j\leq 1$
for $1\leq i,j\leq m$ with $i\ne j$.  for any $\delta+1$ repair sets
$A_{j_1},A_{j_2},\cdots,A_{j_{\delta+1}}$ we
have $$\abs{\parenv{\bigcup_{1\leq i\leq \delta+1}A_{j_i}}}\geq
\sum_{1\leq i\leq \delta+1}|A_{j_i}|-\binom{\delta+1}{2}.$$ Thus, for
any $E_{j_i}\subseteq A_{j_i}$ with $|E_{j_i}|\geq \delta$ for $1\leq
i\leq \delta+1$, we have
$$\abs{\parenv{\bigcup_{1\leq i\leq \delta+1}E_{j_i}}}\geq \sum_{1\leq
  i\leq \delta+1}|E_{j_i}|-\binom{\delta+1}{2}\geq
(\delta+1)\delta-\binom{\delta+1}{2},$$ which means that we need at
least $\binom{\delta+1}{2}-y-1$ erasures before we achieve $\delta+1$
repair sets each of which contains more than $\delta+1$
erasures. Thus, the code $\cC$ can recover from any erasure pattern of
$y$ columns from the first $m$ columns and any other
$\binom{\delta+1}{2}-y-1$ erasures if $h+\delta-1-y\geq
\binom{\delta+1}{2}-y-1>0$.
\end{IEEEproof}

\begin{example}
Set $n=24$, $k=14$, $\delta=2$, $r=2$, and $h=3$. Let $\cA=\{A_i~:~
A_i\triangleq \{3,6,5\}+i\subseteq \Z_7, i\in \Z_7\}$.  According to
Construction \ref{cons_array}, we can modify the code in Example
\ref{example_LRC} into a $3\times 8$ array code, whose parity-check
matrix can be given as:
\begin{equation*}
H=\parenv{\begin{array}{ccc|ccc|ccc|ccc|ccc|ccc|ccc|ccc}
 0&  0&  0  & 10&  0&  0  & 7&  0&  0  & 0&  0&  0  & 5&  0&  0  & 0&  0&  0  & 0&  0&  0  & 0&  0&  0  \\
 0&  0&  0  & 0&  0&  0  & 0&  10&  0  & 7&  0&  0  & 0&  0&  0  & 5&  0&  0  & 0&  0&  0  & 0&  0&  0  \\
 0&  0&  0  & 0&  0&  0  & 0&  0&  0  & 0&  10&  0  & 0&  7&  0  & 0&  0&  0  & 5&  0&  0  & 0&  0&  0  \\
 9&  0&  0  & 0&  0&  0  & 0&  0&  0  & 0&  0&  0  & 0&  0&  10  & 0&  3&  0  & 0&  0&  0  & 0&  0&  0  \\
 0&  0&  0  & 0&  9&  0  & 0&  0&  0  & 0&  0&  0  & 0&  0&  0  & 0&  0&  10  & 0&  3&  0  & 0&  0&  0  \\
 0&  9&  0  & 0&  0&  0  & 0&  0&  3  & 0&  0&  0  & 0&  0&  0  & 0&  0&  0  & 0&  0&  10  & 0&  0&  0  \\
 0&  0&  10  & 0&  0&  7  & 0&  0&  0  & 0&  0&  5  & 0&  0&  0  & 0&  0&  0  & 0&  0&  0  & 0&  0&  0  \\
 6&  1&  0  & 0&  5&  2  & 5&  0&  3  & 2&  0&  8  & 10&  7&  0  & 7&  1&  0  & 1&  3&  0  & 10&  0&  0  \\
 3&  4&  0  & 0&  3&  9  & 1&  0&  2  & 8&  0&  5  & 4&  1&  0  & 5&  3&  0  & 9&  6&  0  & 0&  10&  0  \\
 3&  7&  0  & 0&  8&  4  & 2&  0&  2  & 10&  0&  2  & 6&  3&  0  & 7&  7&  0  & 6&  8&  0  & 0&  0&  10
\end{array}
}.
\end{equation*}
Verified by a computer program, the array code can recover any from
any $2$ column erasures from the first $7$ columns, which is
consistent with the result in Theorem \ref{thm_array}. Note that this
kind of erasure pattern is beyond the minimum Hamming distance $d=5$
as shown in Example \ref{example_LRC}.
\end{example}

By Construction~\ref{cons_array}, we can generate codes that may
recover from some special erasure patterns beyond the minimum Hamming
distance. However, all those erasure patterns do not treat columns
equally, and distinguish between two types of columns. If this is an
unwanted feature, we may arrange the global parity checks across
columns, as done in the following construction.

\begin{construction}\label{cons_rearr}
Let $S$ be a $h$-subset of $\F_q$ and let $(X\subseteq \F_q\setminus
S,\cA=\{A_i~:~1\leq i\leq \ell+1\})$ be a $t$-regular
$(m,r+\delta-1,1)$-packing, where $A_i=\{\theta_{i,j}~:~1\leq j\leq
r+\delta-1\}$ for $1\leq i\leq \ell+1$.  Let $n=v|X|=v\rho$ with $v\geq t$,
then
based on $\cA$ and $S$, we can generate a locally repairable code
$\cC$ according to Construction~\ref{cons_poly}. List the elements of
$X$ as $(x_1,x_2,\cdots,x_{\rho})$.  Define column  vectors $V_{x_a}\in
\F^v_q$  for $a\in [\rho]$ as
\begin{equation*}
V^\intercal_{x_a}=(c_{i_{x_a,1},j_{x_a,1}},c_{i_{x_a,2},j_{x_a,2}},\dots,c_{i_{x_a,t},j_{x_a,t}},c_{\ell+2,(a-1)h/\rho+1},c_{\ell+2,(a-1)h/\rho+2},\cdots,c_{\ell+2,ah/\rho}),
\end{equation*}
where
\begin{equation}\label{eqn_array_re}
\theta_{i_{x_a,b},j_{x_a,b}}=x_a,\,\,\text{for}\,\,1\leq b\leq t.
\end{equation}
\end{construction}
\begin{remark}
In Construction~\ref{cons_rearr}, the fact that
$(X\subseteq\F_q\setminus S,\cA=\{A_i~:~1\leq i\leq \ell+1\})$ is a
regular packing means that $n-h=t\rho$. Thus, by $n=v\rho$, we have
$\rho\mid h$ and $v=t+h/\rho$. This is to say the array given by
\eqref{eqn_array_re} is well defined.
\end{remark}

\begin{corollary}
Let $\cC$ be the $v\times \rho$ array code generated by
Construction~\ref{cons_rearr}.  Then $\cC$ has
$(r,\delta)_i$-locality. If $h\leq \delta^2$, then the code can
recover from any $h+\delta-1$ erasures. Furthermore:
\begin{itemize}
  \item [(I)] The code $\cC$ can recover from any erasure pattern of
    $y\leq 2$ columns and any other $h-y(v-t+1)-1$ erasures.
  \item [(II)] If $\binom{y}{2}\leq \delta$, then the code $\cC$ can
    recover from any erasure pattern of $y$ columns and any other
    $h-2-\binom{y}{2}-y(v-t)$ erasures.
\item [(III)] The code $\cC$ can recover from any erasure pattern of
  $y<\frac{(\delta+1)\delta}{2}-1$ columns and any other
  $\min\{\frac{(\delta+1)\delta}{2}-y(v-t+1)-1,h+\delta-1-y(v-t+1)\}$
  erasures.
\end{itemize}
\end{corollary}
\begin{IEEEproof}
Note that any $y$ columns of $\cC$ can be regarded as $y$ columns from
the first $m$ columns and $y(v-t)$ erasures (sectors) from the global check
symbols, for the code generated by Construction~\ref{cons_array}.
Thus, the desired results follows directly from
Theorem~\ref{thm_array}, respectively.
\end{IEEEproof}

For the case $r\nmid k$ and $h=r-v$, we may modify
Construction~\ref{cons_rearr} as follows.

\begin{construction}\label{cons_rearr_h=r-v}
Let $S$ be an $(r-v)$-subset of $\F_q$ and let $(X\subseteq
\F_q\setminus S,\cB=\{B_i~:~1\leq i\leq \ell+1\})$ be a $t$-regular
$(m,r+\delta-1,1)$-packing. Let $A_i=B_i$ for $1\leq i\leq \ell$ and
$A_{\ell+1}\subseteq B_{\ell+1}$.  Let $n=t|X|=t\rho$ and $k=\ell r+v$, then
based on $\cA$ and $S$, we can generate a locally repairable code
$\cC$ according to Construction~\ref{cons_poly}.  List the elements of
$B_{\ell+1}\setminus A_{\ell+1}$ as $(x_1,x_2,\dots,x_{r-v})$ and $X$
as $(x_1,x_2,\cdots,x_{\rho})$.  Define column vectors $V_{x_a}\in
\F^v_q$ for $a\in [\rho]$ as
\begin{equation*}
V^\intercal_{x_a}=
\begin{cases}
(c_{i_{x_a,1},j_{x_a,1}},c_{i_{x_a,2},j_{x_a,2}},\dots,c_{i_{x_a,t-1},j_{x_a,t-1}},c_{\ell+2,a}),& \text{ if } 1\leq a\leq r-v,\\
(c_{i_{x_a,1},j_{x_a,1}},c_{i_{x_a,2},j_{x_a,2}},\dots,c_{i_{x_a,t},j_{x_a,t}}), & \text{ otherwise,}
\end{cases}
\end{equation*}
where $\theta_{i_{x_a,b},j_{x_a,b}}=x_a$, $1\leq b\leq t-1$ for $1\leq a\leq r-v$
and $1\leq b\leq t$ for $r-v+1\leq a\leq \rho$.
\end{construction}

\begin{corollary}\label{coro_h=r-v}
Let $\cC$ be the $t\times \rho$ array code generated by
Construction~\ref{cons_rearr_h=r-v}.  Then $\cC$ has
$(r,\delta)_i$-locality. If $h\leq \delta^2$, then the code can
recover any $h+\delta-1$ erasures.  Furthermore:
\begin{itemize}
  \item [(I)] The code $\cC$ can recover from any erasure pattern of
    $y\leq 2$ columns and any other $h-2y-1$ erasures.
  \item [(II)] If $\binom{y}{2}\leq \delta$, then the code $\cC$ can
    recover from any erasure pattern of $y$ columns and any other
    $h-2-\binom{y}{2}-y$ erasures.
  \item [(III)] The code $\cC$ can recover from any erasure pattern of
    $y<\frac{(\delta+1)\delta}{2}-1$ columns and any other
    $\min\{\frac{(\delta+1)\delta}{2}-2y-1,h+\delta-1-2y\}$ erasures.
\end{itemize}
\end{corollary}
\begin{IEEEproof}
Note that any $y$ columns of $\cC$ can be regarded as $y$ columns from
the first $m$ columns and at most $y$ erasures (sectors) from the global check
symbols, for the code generated by
Construction~\ref{cons_array}. Thus, the desired results follows
directly from Theorem~\ref{thm_array}, respectively.
\end{IEEEproof}

Based on known results about regular packings, we derive some
parameters of GSD codes resulting from our constructions.  In
particular, we use  two well known classes of Steiner systems that
are the affine geometries and projective geometries.

\begin{lemma}[\cite{CD}]
  \label{lem:steiner_ag}
Let $\beta\geq 2$ be an integer and $q_1$ a prime power, then there
exists a $(2,q_1,q_1^\beta)$-Steiner system.
\end{lemma}

Based on affine geometries and Construction~\ref{cons_rearr_h=r-v}, we have
the following conclusion for GSD codes.

\begin{corollary}\label{coro_GSD_AG}
Let $\beta\geq 2$ be an integer and $q_1$ a prime power.  Set
$q_1=r+\delta-1$, $n=\frac{q_1^{\beta}(q_1^{\beta}-1)}{q_1-1}$, $\delta\geq 2$,
$k=(\frac{q_1^{\beta-1}(q_1^{\beta}-1)}{q_1-1}-1)r+v$
with $1\leq v\leq r-1$, and $h=r-v=q_1-\delta-v+1$.  Let $\cC$ be the
$\frac{q_1^\beta-1}{q_1-1}\times q_1^{\beta}$ array code generated by
Construction~\ref{cons_rearr_h=r-v} using a $(2,q_1,q_1^{\beta})$-Steiner system
from Lemma~\ref{lem:steiner_ag}.  If $h\leq \delta^2$, then the code
$\cC$ is an $[n,k,h+\delta-1]_{q}$ optimal locally repairable
code with $(r,\delta)_i$-locality, where $q\geq q_1^{\beta}+h$. Furthermore:
\begin{itemize}
  \item [(I)] If $y\leq 2$ and $y(\frac{q_1^\beta-1}{q_1-1}-2)>\delta$, then the code $\cC$ is
    a $(y,h-2y-1)$-GSD code.
  \item [(II)] If $\binom{y}{2}\leq \delta$ and
    $y\frac{q_1^\beta-1}{q_1-1}-1-\binom{y}{2}-y>\delta$, then the code $\cC$ is a
    $(y,h-2-\binom{y}{2}-y)$-GSD code.
  \item [(III)] If $y<\frac{(\delta+1)\delta}{2}-1$ and
    $y\frac{q_1^\beta-1}{q_1-1}+\gamma>h+\delta-1$, then the code $\cC$ is a $(y,\gamma)$-GSD
    code, where
    $\gamma=\min\{\frac{(\delta+1)\delta}{2}-2y-1,h+\delta-1-2y\}$
    erasures.
\end{itemize}
Herein, we highlight that the second restriction of each item
comes from the requirement in Definition \ref{def_GSD}-(II).
\end{corollary}
\begin{IEEEproof}
  The proof follows directly from Corollary~\ref{coro_h=r-v},
  Lemma~\ref{lem:steiner_ag}, and Definition~\ref{def_GSD}.
\end{IEEEproof}

\begin{lemma}[\cite{CD}]
  \label{lem:steiner_pg}
Let $\beta\geq 2$ be an integer and $q_1$ a prime power, then there
exists a $(2,q_1+1,\frac{q_1^{\beta+1}-1}{q_1-1})$-Steiner system.
\end{lemma}

Based on projective geometries and
Construction~\ref{cons_rearr_h=r-v}, we have the following conclusion
for GSD codes.

\begin{corollary}\label{coro_GSD_PG}
Let $\beta\geq 2$ be an integer and $q_1$ a prime power.  Set
$q_1+1=r+\delta-1$, $n=\frac{(q_1^{\beta+1}-1)(q_1^\beta-1)}{(q_1-1)^2}$, $\delta\geq 2$,
$k=(\frac{(q_1^{\beta+1}-1)(q_1^\beta-1)}{(q_1-1)(q_1^2-1)}-1)r+v$ with $1\leq v\leq r-1$, and
$h=r-v=q_1-\delta-v+2$.  Let $\cC$ be the $\frac{q_1^{\beta}-1}{q_1-1}\times \frac{q_1^{\beta+1}-1}{q_1-1}$ array code
generated by Construction~\ref{cons_rearr_h=r-v} using a
$(2,q_1+1,\frac{q_1^{\beta+1}-1}{q_1-1})$-Steiner system from
Lemma~\ref{lem:steiner_ag}.  If $h\leq \delta^2$, then the code $\cC$
is an $[n,k,h+\delta-1]_{q}$ optimal locally repairable code with
$(r,\delta)_i$-locality, where $q\geq \frac{q_1^{\beta+1}-1}{q_1-1}+h$
is a prime power. Furthermore:
\begin{itemize}
  \item [(I)] If $y\leq 2$ and $y(\frac{q_1^{\beta}-1}{q_1-1}-2)>\delta$, then the code $\cC$ is
    a $(y,h-2y-1)$-GSD code.
  \item [(II)] If $\binom{y}{2}\leq \delta$ and
    $y\frac{q_1^{\beta}-1}{q_1-1}-1-\binom{y}{2}-y>\delta$, then the code $\cC$ is a
    $(y,h-2-\binom{y}{2}-y)$-GSD code.
  \item [(III)] If $y<\frac{(\delta+1)\delta}{2}-1$ and
    $y\frac{q_1^{\beta}-1}{q_1-1}+\gamma>h+\delta-1$, then the code $\cC$ is a $(y,\gamma)$-GSD
    code, where
    $\gamma=\min\{\frac{(\delta+1)\delta}{2}-2y-1,h+\delta-1-2y\}$
    erasures.
\end{itemize}
\end{corollary}
\begin{IEEEproof}
  The proof follows directly from Corollary~\ref{coro_h=r-v},
  Lemma~\ref{lem:steiner_pg}, and Definition~\ref{def_GSD}.
\end{IEEEproof}

\begin{lemma}[\cite{CD}]\label{SG}
For any $\beta\geq 2$ and prime power $q_1$, there exists a $(3,q_1+1,q_1^\beta+1)$-Steiner system.
\end{lemma}

Similarly, based on Steiner systems from spherical geometries and
Construction~\ref{cons_rearr_h=r-v}, we have the following conclusion
for GSD codes.

\begin{corollary}\label{coro_GSD_SG}
Let $\beta\geq 2$ be an integer and $q_1$ a prime power.  Set
$q_1+1=r+\delta-1$, $n=(q_1^{\beta}+1)\frac{\binom{q_1^{\beta}}{2}}{\binom{q_1}{2}}$, $\delta\geq 2$,
$k=(\frac{(q_1^{\beta}+1)\binom{q_1^{\beta}}{2}}{(q_1+1)\binom{q_1}{2}}-1)r+v$ with $1\leq v\leq r-1$, and
$h=r-v=q_1-\delta-v+2$.  Let $\cC$ be the $\frac{\binom{q_1^{\beta}}{2}}{\binom{q_1}{2}}\times (q_1^\beta+1)$ array code
generated by Construction~\ref{cons_rearr_h=r-v} using a
$(2,q_1+1,\frac{q_1^{\beta+1}-1}{q_1-1})$-Steiner system from
Lemma~\ref{lem:steiner_ag}.  If $h\leq \delta^2$, then the code $\cC$
is an $[n,k,h+\delta-1]_{q}$ optimal locally repairable code with
$(r,\delta)_i$-locality, where $q\geq q^{\beta}_1+1+h$
is a prime power. Furthermore:
\begin{itemize}
  \item [(I)] If $y\leq 2$ and $y(\frac{\binom{q_1^{\beta}}{2}}{\binom{q_1}{2}}-2)>\delta$, then the code $\cC$ is
    a $(y,h-2y-1)$-GSD code.
  \item [(II)] If $\binom{y}{2}\leq \delta$ and
    $y\frac{\binom{q_1^{\beta}}{2}}{\binom{q_1}{2}}-1-\binom{y}{2}-y>\delta$, then the code $\cC$ is a
    $(y,h-2-\binom{y}{2}-y)$-GSD code.
  \item [(III)] If $y<\frac{(\delta+1)\delta}{2}-1$ and
    $y\frac{\binom{q_1^{\beta}}{2}}{\binom{q_1}{2}}+\gamma>h+\delta-1$, then the code $\cC$ is a $(y,\gamma)$-GSD
    code, where
    $\gamma=\min\{\frac{(\delta+1)\delta}{2}-2y-1,h+\delta-1-2y\}$
    erasures.
\end{itemize}
\end{corollary}
\begin{IEEEproof}
  The proof follows directly from Corollary~\ref{coro_h=r-v},
  Lemma~\ref{SG}, and Definition~\ref{def_GSD}.
\end{IEEEproof}

For regular packings, we have the following lemma due to a recursive
construction from \cite{J}. A direct construction is also supplied for
the reader's convenience in Appendix~\ref{app:pack}.

\begin{lemma}[\cite{J}]
Let $n_2=p_1^{m_1}p_2^{m_2}\cdots p_u^{m_u}$, where $p_i$ are distinct
primes, and $m_i>0$, for all $i$. If
$e|\gcd(p^{m_1}_1-1,p^{m_2}_2-1,\cdots,p^{m_u}_u-1)$, then there
exists a $\frac{1}{e^u}\prod_{1\leq i\leq u} (p^{m_i}_1-1)$-regular
$(en_2,e,1)$-packing.
\end{lemma}

Based on regular packings we can also generate GSD codes as follows.

\begin{corollary}\label{coro_GSD_RP}
Let $n_2=p_1^{m_1}p_2^{m_2}\cdots p_u^{m_u}$ and
$e|\gcd(p^{m_1}_1-1,p^{m_2}_2-1,\cdots,p^{m_u}_u-1)$, where $p_i$ are
distinct primes, and $m_i>0$, for all $i$. Define $p= \prod_{1\leq
  i\leq u} (p^{m_i}_1-1)$. Set $e=r+\delta-1$, $n=n_2p/e^{u-1}$,
$\delta\geq 2$, $k=(n_2p/e^u-1)r+v$ with $1\leq v\leq r-1$, and
$h=r-v=e-\delta+1-v$.  Let $\cC$ be the $p/e^u\times en_2$ array code
generated by Construction~\ref{cons_rearr_h=r-v}.  If $h\leq \delta^2$, then
the code $\cC$ is an $[n,k,h+\delta-1]_{q}$ optimal locally
repairable code with $(r,\delta)_i$-locality, where $q\geq en_2+h$ is
a prime power. Furthermore:
\begin{itemize}
  \item [(I)] If $y\leq 2$ and $y(p/e^{u-1}-2)>\delta$, then the code
    $\cC$ is a $(y,h-2y-1)$-GSD code.
  \item [(II)] If $\binom{y}{2}\leq \delta$ and
    $yp/e^{u-1}-1-\binom{y}{2}-y>\delta$, then the code $\cC$ is a
    $(y,h-2-\binom{y}{2}-y)$-GSD code.
\item [(III)] If $y<\frac{(\delta+1)\delta}{2}-1$ and
  $yp/e^{u-1}+\gamma>h+\delta-1$, then the code $\cC$ is a
  $(y,\gamma)$-GSD code, where
  $\gamma=\min\{\frac{(\delta+1)\delta}{2}-2y-1,h+\delta-1-2y\}$
  erasures.
\end{itemize}
\end{corollary}

\begin{table}
  \caption{A comparison of MR-codes, SD-codes, and GSD-codes}
  \label{tab:comp}
  \begin{center}
    \begin{tabular}{|c|c|c|c|c|c|}
      \hline
      $s$ & $\gamma$ & Type & $n$ & Ref. & Comment \\
      \hline\hline
      any & $1$ & MR & $\Theta(q)$ & \cite{BHH} & \\
      $1,2,3$ & $2$ & SD & $\Theta(q)$ & \cite{PB} &\\
      any & $2$ & SD & $\Theta(q)$ & \cite{BPSY} &\\
      any & $3$ & MR & $\Theta(q^{1/3})$ & \cite{GGY} &\\
      $1$ & any & MR & $\Theta(q^{1/(\gamma-1)})$ & \cite{GHJY} & $r|(k+\gamma)$ \\
      any & any & MR & $\Theta(\log q)$ & \cite{CK} & $s=\delta-1$ \\
      any & any & MR & $\Theta(q^{1/\gamma})$ & \cite{GYBS} & $s=\delta-1$ \\
      any & any & MR & $\Theta(q^{1/r})$ & \cite{MK} & $s=\delta-1$ \\
      any & any & GSD & $\Theta(q^{2-1/\beta})$ & Cor.~\ref{coro_GSD_AG},\ref{coro_GSD_PG} & $\beta\geq 2$ is a constant \\
      any & any & GSD & $\Theta(q^{3-2/\beta})$ & Cor.~\ref{coro_GSD_SG} & $\beta\geq 2$ is a constant\\
      any & any & GSD & $\Theta(q^{2})$ & Cor.~\ref{coro_GSD_AG},\ref{coro_GSD_PG} & $q_1$ is a constant \\
      any & any & GSD & $\Theta(q^{3})$ & Cor.~\ref{coro_GSD_SG} & $q_1$ is a constant\\
      any & any & GSD & $\Theta(q^2)$ & Cor.~\ref{coro_GSD_RP} & $e,u,h$ are constants\\
      any & any & GSD & $\Theta(q^{\tau})$ & Cor.~\ref{coro_h=r-v} & Remark~\ref{rem:agpg} (non-explicit) \\
      any & any & GSD & $\Theta(q^{\tau})$ & Cor.~\ref{coro_h=r-v} & Remark~\ref{rem:rp} (non-explicit)\\
      \hline
    \end{tabular}
  \end{center}
\end{table}

Table~\ref{tab:comp} lists some known results about SD codes and MR
code (PMDS codes) as a comparison with the GSD codes we have
constructed. The main point of comparison is the asymptotics of the
length of the code with respect to the field size. For this table,
$n=m(r+\delta-1)$ is the total number of sectors for a codeword, $k$
is the number of sectors for information symbols, $r+\delta-1$ is the
number of columns (i.e., the code has $(r,\delta)_a$-locality), and
$q$ is the field size. For a fair comparison with our results in
Corollaries \ref{coro_GSD_AG}--\ref{coro_GSD_RP}, we consider $r$,
$\delta$, and $\gamma$, as constants when we consider the relationship
between $n$ and $q$. We further make the following remarks:

\begin{remark}
  \label{rem:agpg}
  By Corollaries~\ref{coro_GSD_AG} and~\ref{coro_GSD_PG}, there exist
GSD codes with $n=\Theta(q^2)$, where $h$, $r$, and $\delta$ are regarded as constants, i.e., $q_1$ is a constant,
as already written in Table~\ref{tab:comp}.
Note that if we regard $\beta\geq 2$ as a
constant then $n=\Theta(q^{\frac{2\beta-1}{\beta}})$ with $q=\Theta(q_1^{\beta})$.
In addition, for general cases by
using Steiner systems with parameters $(\tau,r+\delta-1,n_1)$,
Steiner systems are capable of yielding optimal locally repairable
codes (similarly, GSD codes) with length $n=\Theta(q^\tau)$ as shown
in Corollary \ref{coro_steiner} and Remark
\ref{coro_steiner}. Here we apply the fact that Steiner systems are regular packings,
which means that the locally repairable codes in Corollary \ref{coro_steiner} and
Remark \ref{coro_steiner} can yield GSD codes by Construction \ref{cons_rearr_h=r-v} and
Corollary \ref{coro_h=r-v}.
For example, in Corollary \ref{coro_GSD_SG}, we have
$n=\Theta(q^3)$ for the case $\tau=3$, where $q_1+1=r+\delta-1$ is regarded as a constant.
However, the problem of constructing Steiner
systems in general is widely open in combinatorics. For a summary of
combinatorial designs and linear codes, the reader may refer to
\cite{D} for example.
\end{remark}

\begin{remark}
  \label{rem:rp}
According to Corollary~\ref{coro_GSD_RP}, there exist GSD codes with
$n=\frac{\prod_{1\leq i\leq u}(p^{m_i}_i(p^{m_i}_i-1))}{e^u}=\Theta(q^2)$,
where $q\geq e\prod_{1\leq i\leq u}p^{m_i}_i+h$ and we consider
$e=r+\delta-1$, $u$, and $h$ as
constants. In particular, for the case $\delta=2$, this code has
order-optimal length with respect to the bound in Theorem
\ref{theorem_bound_delta>2}. Similarly, for $\tau>2$, to generate
codes with length $n=\Theta(q^\tau)$ we need
regular $\tau$-$(n_1,r+\delta-1,1)$-packings with $\tau>2$ (see Definition \ref{def_packing}),
where we also need to apply Construction \ref{cons_rearr_h=r-v} to rearrange
the locally repairable codes into GSD codes.
\end{remark}

\begin{example}
Set $n=9\times 73=657$, $k=7\times 72+1=505$, $\delta=3$, $r=7$, and
$h=6$. Let $\cA=\{A_i~:~ i\in[73]\}$ be a $(2,9,73)$-Steiner system.
According to Construction \ref{cons_rearr_h=r-v}, we can generate a
$9\times 73$ array code, which forms a $(2,1)$-GSD code (or a
$(1,3)$-GSD code). This code is an optimal $[657,505,9]_{q\geq 79}$
locally repairable code with $(7,3)_i$-locality when viewed as a one
dimensional linear code.
\end{example}

\section{Locally Repairable Codes via Classical Goppa Codes}\label{sec-Goppa}

In this section, inspired by the classical Goppa code \cite{G70}, we
apply a similar method to construct locally repairable codes.

\begin{construction}\label{cons_Goppa}
Let $G_1(x)$ and $G_2(x)$ be two polynomials over $\F_{q^m}$ with
degree $\delta-1$ and $h$, respectively. Let
$S=(\gamma_1,\gamma_2,\cdots,\gamma_n)$ be a sequence of length $n$
over $\F_{q^m}$. Also, let $S_1,\dots,S_{\ell+1}\subseteq\F_{q^m}$ be
subsets such that $|S_i|=r+\delta-1$ for $1\leq i\leq \ell$,
$|S_{\ell+1}|\leq h$, as well as,
\begin{equation*}
\bigcup_{1\leq i\leq \ell+1}S_{i}= \{\gamma_i~:~1\leq i\leq n\}.
\end{equation*}
and
\begin{equation*}
G_1(\gamma_i)G_2(\gamma_i)\ne 0 \quad \text{for }1\leq i\leq n.
\end{equation*}
Define the code $\Gamma_{q^m}(\cS=\{S,S_1,\dots,S_\ell\},\cG=\{G_1,G_2\})$ as a set of vectors $V=(v_1,v_2,\dots,v_n)\in \F^n_{q^m}$
such that
\begin{equation*}
\sum_{1\leq j\leq r+\delta-1}\frac{v_{i,j}}{x-\gamma_{i,j}}\equiv 0\pmod{G_1(x)}\quad \text{for }1\leq i\leq \ell
\end{equation*}
and
\begin{equation*}
\sum_{1\leq j\leq n}\frac{v_{j}}{x-\gamma_{j}}\equiv 0\pmod{G_2(x)},
\end{equation*}
where for $1\leq i\leq \ell$ and $1\leq j\leq r+\delta-1$, we
denote $v_{(i-1)(r+\delta-1)+j}$ as $v_{i,j}$, and
$\gamma_{(i-1)(r+\delta-1)+j}$ as $\gamma_{i,j}$, and
$S_i=\{\gamma_{i,j}\}_{j=1}^{r+\delta-1}$.
\end{construction}

\begin{lemma}\label{lemma_Gamma_k}
The code $\Gamma_{q^m}(\cS,\cG)$ generated by
Construction~\ref{cons_Goppa} is an $[n,k]_{q^m}$ linear code with
$k\geq n-\ell(\delta-1)-h$, whose code symbol $v_{i,j}$ has
$(r,\delta)$-locality for $1\leq i\leq \ell$ and $1\leq j\leq
r+\delta-1$.
\end{lemma}
\begin{IEEEproof}
By the properties of classical Goppa codes (refer to \cite{MS} Chapter
12.3 for more details), the parity-check matrix of
$\Gamma_{q^m}(\cS,\cG)$ may be written as
\begin{equation*}
P=
\begin{pmatrix}
P_{1,1}&0&\cdots&0&0\\
0&P_{1,2}&\cdots&0&0\\
\vdots&\vdots&\ddots&\vdots&\vdots\\
0&0&\cdots&P_{1,\ell}&0\\
P_{2,1}&P_{2,2}&\cdots&P_{2,\ell}&P_{2,\ell+1}\\
\end{pmatrix},
\end{equation*}
where
\begin{equation*}
P_{1,i}=\parenv{
\begin{matrix}
G_1^{-1}(\gamma_{i,1})&G_1^{-1}(\gamma_{i,2})&G_1^{-1}(\gamma_{i,3})&\cdots&G_1^{-1}(\gamma_{i,r+\delta-1})\\
G_1^{-1}(\gamma_{i,1})\gamma_{i,1}&G_1^{-1}(\gamma_{i,2})\gamma_{i,2}&G_1^{-1}(\gamma_{i,3})\gamma_{i,3}&\cdots&G_1^{-1}(\gamma_{i,r+\delta-1})\gamma_{i,r+\delta-1}\\
\vdots&\vdots&\vdots& &\vdots\\
G_1^{-1}(\gamma_{i,1})\gamma^{\delta-3}_{i,1}&G_1^{-1}(\gamma_{i,2})\gamma^{\delta-3}_{i,2}&G_1^{-1}(\gamma_{i,3})\gamma^{\delta-3}_{i,3}&\cdots
&G_1^{-1}(\gamma_{i,r+\delta-1})\gamma^{\delta-3}_{i,r+\delta-1}\\
G_1^{-1}(\gamma_{i,1})\gamma^{\delta-2}_{i,1}&G_1^{-1}(\gamma_{i,2})\gamma^{\delta-2}_{i,2}&G_1^{-1}(\gamma_{i,3})\gamma^{\delta-2}_{i,3}&\cdots
&G_1^{-1}(\gamma_{i,r+\delta-1})\gamma^{\delta-2}_{i,r+\delta-1}\\
\end{matrix}}
\end{equation*}
for $1\leq i\leq \ell$, and
\begin{equation*}
\begin{split}
&\parenv{
\begin{matrix}
P_{2,1}&P_{2,2}&P_{2,3}&\cdots&P_{2,\ell}&P_{2,\ell+1}\\
\end{matrix}}\\
=&\parenv{
\begin{matrix}
G_2^{-1}(\gamma_{1})&G_2^{-1}(\gamma_{2})&G_2^{-1}(\gamma_{3})&\cdots&G_2^{-1}(\gamma_{n})\\
G_2^{-1}(\gamma_{1})\gamma_{1}&G_2^{-1}(\gamma_{2})\gamma_{2}&G_2^{-1}(\gamma_{3})\gamma_{3}&\cdots&G_2^{-1}(\gamma_{n})\gamma_{n}\\
\vdots&\vdots&\vdots& &\vdots\\
G_2^{-1}(\gamma_{1})\gamma^{h-2}_{1}&G_2^{-1}(\gamma_{2})\gamma^{h-2}_{2}&G_2^{-1}(\gamma_{3})\gamma^{h-2}_{3}&\cdots
&G_2^{-1}(\gamma_{n})\gamma^{h-2}_{n}\\
G_2^{-1}(\gamma_{1})\gamma^{h-1}_{1}&G_2^{-1}(\gamma_{2})\gamma^{h-1}_{2}&G_2^{-1}(\gamma_{3})\gamma^{h-1}_{3}&\cdots
&G_2^{-1}(\gamma_{n})\gamma^{h-1}_{n}\\
\end{matrix}}
\end{split}
\end{equation*}
and in particular,
\begin{equation*}
P_{2,i}=\parenv{
\begin{matrix}
G_2^{-1}(\gamma_{i,1})&G_2^{-1}(\gamma_{i,2})&G_2^{-1}(\gamma_{i,3})&\cdots&G_2^{-1}(\gamma_{i,r+\delta-1})\\
G_2^{-1}(\gamma_{i,1})\gamma_{i,1}&G_2^{-1}(\gamma_{i,2})\gamma_{i,2}&G_2^{-1}(\gamma_{i,3})\gamma_{i,3}&\cdots&G_2^{-1}(\gamma_{i,r+\delta-1})\gamma_{i,r+\delta-1}\\
\vdots&\vdots&\vdots& &\vdots\\
G_2{-1}(\gamma_{i,1})\gamma^{\delta-3}_{i,1}&G_2^{-1}(\gamma_{i,2})\gamma^{\delta-3}_{i,2}&G_2^{-1}(\gamma_{i,3})\gamma^{\delta-3}_{i,3}&\cdots
&G_2^{-1}(\gamma_{i,r+\delta-1})\gamma^{\delta-3}_{i,r+\delta-1}\\
G_2^{-1}(\gamma_{i,1})\gamma^{\delta-2}_{i,1}&G_2^{-1}(\gamma_{i,2})\gamma^{\delta-2}_{i,2}&G_2^{-1}(\gamma_{i,3})\gamma^{\delta-2}_{i,3}&\cdots
&G_2^{-1}(\gamma_{i,r+\delta-1})\gamma^{\delta-2}_{i,r+\delta-1}\\
\end{matrix}}
\end{equation*}
for $1\leq i\leq \ell$. Thus, by the fact $G_1(\gamma_{i,j})\ne 0$ for
$1\leq i\leq \ell$ and $1\leq j\leq r+\delta-1$, we have the code
symbol $v_{i,j}$ has $(r,\delta)$-locality, i.e., the matrix $P_{1,i}$
is a parity-check matrix of a code with minimum Hamming distance at
least $\delta$. Now the desired result follows from the fact that the
code determined by $P$ has parameters $[n,k\geq
  n-h-\ell(\delta-1)]_{q^m}$.
\end{IEEEproof}

To bound the Hamming distance of $\Gamma_{q^m}(\cS,\cG)$, we define an
auxiliary code over the splitting field of $G_1(x)G_2(x)$. Let
$\F_{q^{m_1}}$ be the splitting field of $G_1(x)G_2(x)$ and let
$B_1=\{b_{1,1},b_{1,2},\dots,b_{1,\delta-1}\}$ and
$B_2=\{b_{2,1},b_{2,2},\dots,b_{2,h}\}$ be the roots of $G_1(x)$ and
$G_2(x)$ over $\F_{q^{m_1}}$, respectively.  Define the code
$\Gamma_{q^{m_1}}(\cS,\cG)$ as a set of vectors
$V^*=(v^*_1,v^*_2,\dots,v^*_n)\in \F^n_{q^{m_1}}$ such that
\begin{equation}\label{eqn_v*_G1}
\sum_{1\leq j\leq r+\delta-1}\frac{v^*_{i,j}}{x-\gamma_{i,j}}\equiv 0\pmod{G_1(x)}
\end{equation}
and
\begin{equation}\label{eqn_v*_G2}
\sum_{1\leq j\leq n}\frac{v^*_{j}}{x-\gamma_{j}}\equiv 0\pmod{G_2(x)},
\end{equation}
where for $1\leq i\leq \ell$ and $1\leq j\leq r+\delta-1$, we denote
$v^*_{(i-1)(r+\delta-1)+j}$ as $v^*_{i,j}$.

\begin{lemma}\label{lemma_parity_check}
For the code $\Gamma_{q^{m_1}}(\cS,\cG)$, if $G_1(x)G_2(x)$ has $\delta-1+h$ distinct roots
over $\F_{q^{m_1}}$, then
its parity-check matrix can be written as
\begin{equation*}
P^*=
\begin{pmatrix}
P^*_{1,1}&0&\cdots&0&0\\
0&P^*_{1,2}&\cdots&0&0\\
\vdots&\vdots&\ddots&\vdots&\vdots\\
0&0&\cdots&P^*_{1,\ell}&0\\
P^*_{2,1}&P^*_{2,2}&\cdots&P^*_{2,\ell}&P^*_{2,\ell+1}\\
\end{pmatrix},
\end{equation*}
where for $1\leq i\leq \ell$ $P^*_{1,i}=(p^{(i)}_{t,j})$ is a
$(\delta-1)\times (r+\delta-1)$ Cauchy matrix with
$p^{(i)}_{t,j}=\frac{1}{b_{1,t}-\gamma_{i,j}}$ for $1\leq t\leq
\delta-1$ and $1\leq j\leq r+\delta-1$ and
$(P^*_{2,1},P^*_{2,2},\dots, P^*_{2,\ell+1})=(p_{t,j})$ is an $h\times
n$ Cauchy matrix with $p_{t,j}=\frac{1}{b_{2,t}-\gamma_{j}}$ for
$1\leq t\leq h$ and $1\leq j\leq n$. In particular,
$P^*_{2,i}=(p^{(i)}_{t,j})$ is a $(\delta-1)\times (r+\delta-1)$
Cauchy matrix with $p^{(i)}_{t,j}=\frac{1}{b_{2,t}-\gamma_{i,j}}$ for
$1\leq i\leq \ell$, $1\leq t\leq \delta-1$ and $1\leq j\leq
r+\delta-1$.
\end{lemma}
\begin{proof}
Obviously, if $V^*\in \Gamma_{q^{m_1}}(\cS,\cG)$ is a codeword, then
\eqref{eqn_v*_G1} and \eqref{eqn_v*_G2} imply that $P^*V^*={\bm 0}$.
For any vector $V'\in \F^n_{q^{m_1}}$ with $P^*V'={\bm 0}$, we have
\begin{equation*}
\sum_{1\leq j\leq r+\delta-1}\frac{v'_{i,j}}{x-\gamma_{i,j}}\equiv 0\pmod{x-b_{1,i}}
\end{equation*}
for $1\leq i\leq \delta-1$
and
\begin{equation*}
\sum_{1\leq j\leq n}\frac{v'_{j}}{x-\gamma_{j}}\equiv 0\pmod{x-b_{2,j}},
\end{equation*}
for $1\leq j\leq h$.
Now the fact that $G_1(x)G_2(x)$ has $h+\delta-1$ distinct roots means that
\begin{equation*}
\sum_{1\leq j\leq r+\delta-1}\frac{v'_{i,j}}{x-\gamma_{i,j}}\equiv 0\pmod{G_1(x)=\prod_{1\leq i\leq \delta-1}(x-b_{1,i})}
\end{equation*}
and
\begin{equation*}
\sum_{1\leq j\leq n}\frac{v'_{j}}{x-\gamma_{j}}\equiv 0\pmod{G_2(x)=\prod_{1\leq i\leq h}(x-b_{2,i})},
\end{equation*}
i.e., $V'\in \Gamma_{q^{m_1}}(\cS,\cG)$.
This completes the proof.
\end{proof}

\begin{theorem}\label{theorem_distance}
Assume $G_1(x)G_2(x)$ has $\delta-1+h$ distinct roots over $\F_{q^{m_1}}$.
For any $t+1$-subset $D$ of $[\ell]$, if
\begin{equation}\label{eqn_cond_S_i}
\abs{S_i\cap \parenv{\bigcup_{j\ne i,j\in D} S_j}}\leq \delta-1 \quad \text{for }i\in D
\end{equation}
and
\begin{equation}\label{eqn_cond_ell}
S_{\ell+1}\cap S_i=\varnothing\quad \text{for }1\leq i\leq \ell,
\end{equation}
then the code $\Gamma_{q^{m_1}}(\cS,\cG)$ has minimum Hamming distance
$d\geq \min\{(t+1)\delta,h+\delta\}$.
\end{theorem}

However, before proving the theorem, we first prove two technical
lemmas which will be used in the proof.

\begin{lemma}\label{lemma_Rank_M}
Let $\cE=\{E_1,E_2,\dots, E_{\tau}\}$ be a set of subsets of
$\F_{q^{m_1}}$, and let $\Theta=\{\theta_i~:~1\leq i\leq
h_1\}\subseteq \F_{q^{m_1}}\backslash (\bigcup_{1\leq j\leq\tau}E_j)$.
Define an $h_1\times \tau$ matrix, $M(\cE,\Theta)$, whose $(i,j)$ entry is
\begin{equation}\label{eqn_M_E_Theta}
M(\cE,\Theta)_{i,j}=\frac{1}{f_{E_i}(\theta_j)},
\end{equation}
where
\begin{equation*}
f_{E_i}(x)=\prod_{\theta\in E_i}(x-\theta)\quad \text{for }1\leq i\leq \tau.
\end{equation*}
If, for any $t-1$ sets $E_{i_1},E_{i_2},\cdots, E_{i_{t-1}}\in \cE$,
$|\bigcup_{1\leq j\leq t-1}E_{i_j}|< h_1$, and any $E_i\in\cE$ cannot
be covered by $t\leq \tau-1$ other elements of $\cE$, i.e.,
\begin{equation}\label{eqn_cond_E_i}
E_i\not\subseteq \bigcup_{1\leq j\leq t, i_j\ne i}E_{i_j}\quad \text{for any }\{i_j~:~1\leq j\leq t\}\subseteq [\tau]\backslash \{i\},
\end{equation}
then any $t$ columns of $M(\cE,\Theta)$ are linearly independent over $\F_{q^{m_1}}$.
\end{lemma}
\begin{IEEEproof}
We assume to the contrary that there exist $t$ columns of
$M(\cE,\Theta)$ that are linearly dependent over $\F_{q^{m_1}}$, which
form a sub-matrix of $M(\cE,\theta)$ given by
\begin{equation*}
M'\triangleq \parenv{
\begin{matrix}
\frac{1}{f_{E_{i_1}}(\theta_1)} &\frac{1}{f_{E_{i_2}}(\theta_1)}&\cdots &\frac{1}{f_{E_{i_{t_1}}}(\theta_1)}\\
\frac{1}{f_{E_{i_1}}(\theta_2)} &\frac{1}{f_{E_{i_2}}(\theta_2)}&\cdots &\frac{1}{f_{E_{i_{t_1}}}(\theta_2)}\\
\vdots &\vdots&  &\vdots\\
\frac{1}{f_{E_{i_1}}(\theta_{h_1})} &\frac{1}{f_{E_{i_2}}(\theta_{h_1})}&\cdots &\frac{1}{f_{E_{i_{t_1}}}(\theta_{h_1})}\\
\end{matrix}},
\end{equation*}
where $\Theta\subseteq \F_{q^{m_1}}\backslash (\bigcup_{1\leq j\leq\tau}E_j)$
means that $M'$ is well defined.
Since $\rank(M')<t$, there exists a function
$f(x)=\sum_{1\leq j\leq t_1}e_{i}\frac{1}{f_{E_{i_j}}(x)}$ where
$\theta_{1},\dots,\theta_{h_1}$ are roots of $f(x)=0$, and where
$(e_1,e_2,\dots,e_{t_1})\ne \bm 0$.
Denote $E=\bigcup_{1\leq j\leq t}E_{i_j}$.
It follows that
\begin{equation*}
f^*(x)\triangleq f_{E}(x)\parenv{\sum_{1\leq j\leq t}\frac{e_j}{f_{E_{i_j}}(x)}}=0
\end{equation*}
with degree at most $\max\{|\bigcup_{1\leq j\ne s_1\leq t}E_{i_j}|~:~1\leq s_1\leq t\}<h_1$ has $h_1$ roots over $\F_{q^{m_1}}$,
which means $f^*(x)=0$.
However, by \eqref{eqn_cond_E_i}, for any given $1\leq s_1\leq t$
there exists a $\theta\in \Theta$ such that
$$f_{E\backslash {E_{i_{s_1}}}}(\theta)\ne 0$$
and
$$f_{E\backslash {E_{i_{s_2}}}}(\theta)= 0\quad \text{for all }1\leq s_2\ne s_1\leq t.$$
Thus, $f_{E\backslash {E_{i_{j}}}}(\theta)$ for
$1\leq j\leq t$ are linearly independent
over $\F_{q^{m_1}}$, which is a contradiction with $(e_1,e_2,\dots,e_{t_1})\ne \bm 0$.
\end{IEEEproof}
\begin{remark}
When $\delta=1$, $M(\cE,\Theta)$ is exactly the well-known Cauchy matrix
and the result in Lemma~\ref{lemma_Rank_M} is just the known property of
Cauchy matrices.
\end{remark}

\begin{lemma}\label{lemma_M_dia}
Let $W=(\alpha_1,\alpha_2,\cdots,\alpha_n)\in \F^n_{q^{m_1}}$ and let
$W^*=\{\alpha_i~:~1\leq i\leq n\}$. Denote $W_i=
\{\alpha_{i,j}\triangleq\alpha_{(i-1)(\tau)+j}\,:1\leq j\leq \tau\,\}$
for $1\leq i\leq \frac{n}{\tau}$, where $\tau$ is an integer factor of
$n$. Let $\delta$ be an integer with $\delta\leq m_1$,
$\Theta_1=\{\theta_{1,i}~:~1\leq i\leq \delta-1\}\subseteq
\F_{q^{m_1}}\backslash W^*$, $\Theta_2=\{\theta_{2,i}~:~1\leq i\leq
h\}\subseteq \F_{q^{m_1}}\backslash (W^*\cup \Theta_1)$, and let $M$
be a matrix satisfying
\begin{equation*}
M=\parenv{
\begin{matrix}
M_{1,1}&0&\cdots&0\\
0&M_{1,2}&\cdots&0\\
\vdots&\vdots&\ddots&\vdots\\
0&0&\cdots&M_{1,\frac{n}{\tau}}\\
M_{2,1}&M_{2,2}&\cdots&M_{2,\frac{n}{\tau}}\\
\end{matrix}},
\end{equation*}
where for $1\leq i\leq \frac{n}{\tau}$ $M_{1,i}=(m^{(i)}_{t,j})$ is a
$(\delta-1)\times \tau$ Cauchy matrix with
$m^{(i)}_{t,j}=\frac{1}{\theta_{1,t}-\alpha_{i,j}}$ for $1\leq t\leq
\delta-1$ and $1\leq j\leq \tau$ and $(M_{2,1},M_{2,2},\dots,
M_{2,\frac{n}{\tau}})=(m_{t,j})$ is an $h\times n$ Cauchy matrix with
$m_{t,j}=\frac{1}{\theta_{2,t}-\alpha_{j}}$ for $1\leq t\leq h$ and
$1\leq j\leq n$.  If $|\Theta_1\cup\Theta_2|=h+\delta-1$ and
$|W_i|=\tau$, then the matrix $M$ can be rewritten as
\begin{equation*}
M=LM^*R=\parenv{
\begin{matrix}
L_1&0&\cdots&0&0\\
0&L_2&\cdots&0&0\\
\vdots&\vdots&\ddots&\vdots&\vdots\\
0&0&\cdots&L_{\frac{n}{\tau}}&0\\
0&0&\cdots&0&I_{h}\\
\end{matrix}}\parenv{
\begin{matrix}
M^*_{1,1}&0&\cdots&0\\
0&M^*_{1,2}&\cdots&0\\
\vdots&\vdots&\ddots&\vdots\\
0&0&\cdots&M^*_{1,\frac{n}{\tau}}\\
M^*_{2,1}&M^*_{2,2}&\cdots&M^*_{2,\frac{n}{\tau}}\\
\end{matrix}}
\parenv{
\begin{matrix}
R_1&0&\cdots&0\\
0&R_2&\cdots&0\\
\vdots&\vdots&\ddots&\vdots\\
0&0&\cdots&R_{\frac{n}{\tau}}\\
\end{matrix}},
\end{equation*}
where, for $1\leq i\leq \frac{n}{\tau}$, $|A_i|\ne 0$, $M^*_{1,i}=(I_{\delta-1},0_{(\delta-1)\times (\tau-\delta+1)})$
and $M^*_{2,i}=(M_i,M(\cE_i,\Theta_2))$ with
\begin{equation*}
\cE_i=\{E_{i,j}=\{\alpha_{i,1},\dots,\alpha_{i,\delta-1},\alpha_{i,j}\}~:~\delta\leq j\leq \tau\}
\end{equation*} and
$M(\cE_i,\Theta_2)$ defined by \eqref{eqn_M_E_Theta} and \eqref{lemma_Rank_M}.
\end{lemma}
\begin{IEEEproof}
We prove this lemma by induction on $\delta$.  For the base case we
consider the case $\delta=2$. Note that
$\begin{pmatrix}M_{1,i}\\ M_{2,i}\end{pmatrix}$ is a Cauchy matrix for
$1\leq i\leq \frac{n}{\tau}$.  This, together with the facts
$|\Theta_1\cup\Theta_2|=h+\delta-1$, $|W_i|=\tau$, and $\Theta_1\cup
\Theta_2\subseteq \F_{q^{m_1}}\backslash W^*$, means that
$\begin{pmatrix} M_{1,i}\\ M_{2,i}\end{pmatrix}$ can be rewritten as
\begin{equation*}
\begin{pmatrix}
  M_{1,i}\\ M_{2,i}
\end{pmatrix}
=
\begin{pmatrix}
L^{(2)}_i &0\\
0& I_h\\
\end{pmatrix}
\begin{pmatrix}
  M^{(2)}_{1,i}\\ M^{(2)}_{2,i}
\end{pmatrix}
R^{(1)}_i
\quad \text{for }1\leq i\leq \frac{n}{\tau},
\end{equation*}
where
\begin{equation*}
M^{(2)}_{1,i}=\parenv{
\begin{matrix}
1&0&0&\dots&0\\
\end{matrix}}_{1\times \tau}
\end{equation*}
and
\begin{equation*}
M^{(2)}_{2,i}=\parenv{
\begin{matrix}
\frac{1}{\theta_{2,1}-\alpha_{i,1}}&\frac{1}{(\theta_{2,1}-\alpha_{i,1})(\theta_{2,1}-\alpha_{i,2})}
&\frac{1}{(\theta_{2,1}-\alpha_{i,1})(\theta_{2,1}-\alpha_{i,3})}&\dots
&\frac{1}{(\theta_{2,1}-\alpha_{i,1})(\theta_{2,1}-\alpha_{i,\tau})}\\
\frac{1}{\theta_{2,2}-\alpha_{i,1}}&\frac{1}{(\theta_{2,2}-\alpha_{i,1})(\theta_{2,2}-\alpha_{i,2})}
&\frac{1}{(\theta_{2,2}-\alpha_{i,1})(\theta_{2,2}-\alpha_{i,3})}&\dots
&\frac{1}{(\theta_{2,2}-\alpha_{i,1})(\theta_{2,2}-\alpha_{i,\tau})}\\
\vdots&\vdots&\vdots& &\vdots\\
\frac{1}{\theta_{2,h}-\alpha_{i,1}}&\frac{1}{(\theta_{2,h}-\alpha_{i,1})(\theta_{2,h}-\alpha_{i,2})}
&\frac{1}{(\theta_{2,h}-\alpha_{i,1})(\theta_{2,h}-\alpha_{i,3})}&\dots
&\frac{1}{(\theta_{2,h}-\alpha_{i,1})(\theta_{2,h}-\alpha_{i,\tau})}\\\end{matrix}}
\end{equation*}
and the lemma follows for this case.

For the induction hypothesis we assume that the desired results hold
for $2\leq \delta\leq u$. For the induction step, namely, the case
$\delta=u+1$, similarly, $\Theta_1\cup \Theta_2\subseteq
\F_{q^{m_1}}\backslash W^*$ means that
$\begin{pmatrix}M_{1,i}\\ M_{2,i}\end{pmatrix}$ can be rewritten as
\begin{equation}\label{eqn_M_i}
\begin{pmatrix}
  M_{1,i}\\ M_{2,i}
\end{pmatrix}
=
\begin{pmatrix}
L^{(u+1)}_i &0\\
0& I_h\\
\end{pmatrix}
\begin{pmatrix}
  M^{(u+1)}_{1,i}\\ M^{(u+1)}_{2,i}
\end{pmatrix}
R^{(u)}_i
\quad \text{for }1\leq i\leq \frac{n}{\tau},
\end{equation}
where
\begin{equation*}
M^{(u+1)}_{1,i}=\parenv{
\begin{matrix}
1&0&0&\dots&0\\
0&\frac{1}{\theta_{1,2}-\alpha_{i,2}}&\frac{1}{\theta_{1,2}-\alpha_{i,3}}&\dots&\frac{1}{\theta_{1,2}-\alpha_{i,\tau}}\\
\vdots&\vdots&\vdots& &\vdots\\
0&\frac{1}{\theta_{1,\delta-1}-\alpha_{i,2}}&\frac{1}{\theta_{1,\delta-1}-\alpha_{i,3}}&\dots&\frac{1}{\theta_{1,\delta-1}-\alpha_{i,\tau}}\\
\end{matrix}}=
\parenv{
\begin{matrix}
1&0\\
0&M^{(u)}_{1,i}
\end{matrix}}
\end{equation*}
and
\begin{equation*}
M^{(u+1)}_{2,i}=\parenv{
\begin{matrix}
\frac{1}{\theta_{2,1}-\alpha_{i,1}}&\frac{1}{(\theta_{2,1}-\alpha_{i,1})(\theta_{2,1}-\alpha_{i,2})}
&\frac{1}{(\theta_{2,1}-\alpha_{i,1})(\theta_{2,1}-\alpha_{i,3})}&\dots
&\frac{1}{(\theta_{2,1}-\alpha_{i,1})(\theta_{2,1}-\alpha_{i,\tau})}\\
\frac{1}{\theta_{2,2}-\alpha_{i,1}}&\frac{1}{(\theta_{2,2}-\alpha_{i,1})(\theta_{2,2}-\alpha_{i,2})}
&\frac{1}{(\theta_{2,2}-\alpha_{i,1})(\theta_{2,2}-\alpha_{i,3})}&\dots
&\frac{1}{(\theta_{2,2}-\alpha_{i,1})(\theta_{2,2}-\alpha_{i,\tau})}\\
\vdots&\vdots&\vdots& &\vdots\\
\frac{1}{\theta_{2,h}-\alpha_{i,1}}&\frac{1}{(\theta_{2,h}-\alpha_{i,1})(\theta_{2,h}-\alpha_{i,2})}
&\frac{1}{(\theta_{2,h}-\alpha_{i,1})(\theta_{2,h}-\alpha_{i,3})}&\dots
&\frac{1}{(\theta_{2,h}-\alpha_{i,1})(\theta_{2,h}-\alpha_{i,\tau})}\\\end{matrix}}=
T^{(u)}_i\parenv{
\begin{matrix}
{\bm 1}&M^{(u)}_{2,i}
\end{matrix}}
\end{equation*}
with $T^{(u)}_i=\mathrm{diag}(\frac{1}{\theta_{2,1}-\alpha_{i,1}},\frac{1}{\theta_{2,2}-\alpha_{i,1}},\dots,\frac{1}{\theta_{2,h}-\alpha_{i,1}})$.
By the induction hypothesis,
\begin{equation}\label{eqn_M_u}
\begin{split}
M^{(u)}=&\parenv{
\begin{matrix}
M^{(u)}_{1,1}&0&\cdots&0\\
0&M^{(u)}_{1,2}&\cdots&0\\
\vdots&\vdots&\ddots&\vdots\\
0&0&\cdots&M^{(u)}_{1,\frac{n}{\tau}}\\
M^{(u)}_{2,1}&M^{(u)}_{2,2}&\cdots&M^{(u)}_{2,\frac{n}{\tau}}\\
\end{matrix}}\\
=&\parenv{
\begin{matrix}
L'_1&0&\cdots&0&0\\
0&L'_2&\cdots&0&0\\
\vdots&\vdots&\ddots&\vdots&\vdots\\
0&0&\cdots&L'_{\frac{n}{\tau}}&0\\
0&0&\cdots&0&I_{h}\\
\end{matrix}}\parenv{
\begin{matrix}
M'_{1,1}&0&\cdots&0\\
0&M'_{1,2}&\cdots&0\\
\vdots&\vdots&\ddots&\vdots\\
0&0&\cdots&M'_{1,\frac{n}{\tau}}\\
M'_{2,1}&M'_{2,2}&\cdots&M'_{2,\frac{n}{\tau}}\\
\end{matrix}}
\parenv{
\begin{matrix}
R'_1&0&\cdots&0\\
0&R'_2&\cdots&0\\
\vdots&\vdots&\ddots&\vdots\\
0&0&\cdots&R'_{\frac{n}{\tau}}\\
\end{matrix}}
\end{split}
\end{equation}
where for $1\leq i\leq \frac{n}{\tau}$, $M'_{1,i}=(I_{u},0_{u\times (\tau-\delta+1)})$
and $M'_{2,i}=(M'_i,M(\cE'_i,\Theta_2))$ with
\begin{equation*}
\cE'_i=\{E'_{i,j}=\{\alpha_{i,2},\dots,\alpha_{i,\delta-1},\alpha_{i,j}\}~:~u+1\leq j\leq \tau\}.
\end{equation*}

Combining \eqref{eqn_M_i} and \eqref{eqn_M_u}, we have

\begin{equation*}
\begin{split}
M=&\parenv{
\begin{matrix}
M_{1,1}&0&\cdots&0\\
0&M_{1,2}&\cdots&0\\
\vdots&\vdots&\ddots&\vdots\\
0&0&\cdots&M_{1,\frac{n}{\tau}}\\
M_{2,1}&M_{2,2}&\cdots&M_{2,\frac{n}{\tau}}\\
\end{matrix}}\\
=&\parenv{
\begin{matrix}
L_1&0&\cdots&0&0\\
0&L_2&\cdots&0&0\\
\vdots&\vdots&\ddots&\vdots&\vdots\\
0&0&\cdots&L_{\frac{n}{\tau}}&0\\
0&0&\cdots&0&I_{h}\\
\end{matrix}}\parenv{
\begin{matrix}
M^*_{1,1}&0&\cdots&0\\
0&M^*_{1,2}&\cdots&0\\
\vdots&\vdots&\ddots&\vdots\\
0&0&\cdots&M^*_{1,\frac{n}{\tau}}\\
M^*_{2,1}&M^*_{2,2}&\cdots&M^*_{2,\frac{n}{\tau}}\\
\end{matrix}}
\parenv{
\begin{matrix}
R_1&0&\cdots&0\\
0&R_2&\cdots&0\\
\vdots&\vdots&\ddots&\vdots\\
0&0&\cdots&R_{\frac{n}{\tau}}\\
\end{matrix}}.
\end{split}
\end{equation*}
Here, for $1\leq i\leq \frac{n}{\tau}$, $R_i=R'_{i}R^{(u)}_i$,
$L_i=L^{(u)}_i\parenv{\begin{matrix} 1&0\\ 0&L'_i\\
\end{matrix}}
$, $M^*_{1,i}=\parenv{\begin{matrix}
1&0\\
0&M'_{1,i}\\
\end{matrix}}$
and
$$M^*_{2,i}=T^{(u)}_i\parenv{
{\bm 1},M'_{2,i}
}=\parenv{
T^{(u)}_i{\bm 1},T^{(u)}_iM'_i,T^{(u)}_i M(\cE'_i,\Theta_2)
}=\parenv{
  M_i,M(\cE_i,\Theta_2)}$$
with $M_i=(T^{(u)}_i{\bm 1},T^{(u)}_iM'_i)$ and
\begin{equation*}
\cE_i=\{E_{i,j}=\{\alpha_{i,1},\dots,\alpha_{i,\delta-1},\alpha_{i,j}\}~:~u+1\leq j\leq \tau\}.
\end{equation*}
By induction, this completes the proof.
\end{IEEEproof}

\begin{IEEEproof}[Proof of Theorem~\ref{theorem_distance}]
By Lemma~\ref{lemma_parity_check} the parity-check matrix can be given as
\begin{equation*}
P^*=\parenv{
\begin{matrix}
P^*_{1,1}&0&\cdots&0&0\\
0&P^*_{1,2}&\cdots&0&0\\
\vdots&\vdots&\ddots&\vdots&\vdots\\
0&0&\cdots&P^*_{1,\ell}&0\\
P^*_{2,1}&P^*_{2,2}&\cdots&P^*_{2,\ell}&P^*_{2,\ell+1}\\
\end{matrix}},
\end{equation*}
where for $1\leq i\leq \ell$ $P^*_{1,i}=(p^{(i)}_{u,j})$ is a
$(\delta-1)\times (r+\delta-1)$ Cauchy matrix with
$p^{(i)}_{u,j}=\frac{1}{b_{1,u}-\gamma_{i,j}}$ for $1\leq u\leq
\delta-1$ and $1\leq j\leq r+\delta-1$ and
$(P^*_{2,1},P^*_{2,2},\dots, P^*_{2,\ell+1})=(p_{u,j})$ is an $h\times
n$ Cauchy matrix with $p_{u,j}=\frac{1}{b_{2,u}-\gamma_{j}}$ for
$1\leq u\leq h$ and $1\leq j\leq n$.

We consider the case that there are at most $e\leq\min\{(t+1)\delta-1,h+\delta-1\}$ erasures in
total, i.e., $e=\sum_{1\leq i\leq \ell+1}|E_i|\leq
\min\{t\delta,h+\delta-1\}$.  To bound the Hamming distance we only
need to consider erasure patterns such that $E_i\subseteq S_i$ for
$1\leq i\leq \ell+1$ and $|E_i|\geq \delta$ for $1\leq i\leq \ell$.
Let $P^*(\cE)$ be the sub-matrix formed by the columns corresponding
to $E_i$ $1\leq i\leq \ell+1$, that is the column
$(0,\dots,0,\frac{1}{b_{1,1}-\gamma_{i,j}},\dots,\frac{1}{b_{1,\delta-1}-\gamma_{i,j}},
0\dots,0,\frac{1}{b_{2,1}-\gamma_{i,j}},\dots,\frac{1}{b_{2,h}-\gamma_{i,j}})^\intercal$
is chosen if $\gamma_{i,j}\in E_i\subseteq S_i$. It is easy to check
that $P^*(\cE)$ can be written as
\begin{equation*}
P^*(\cE)=\parenv{
\begin{matrix}
P^{\cE}_{1,i_1}&0&\cdots&0&0\\
0&P^{\cE}_{1,i_2}&\cdots&0&0\\
\vdots&\vdots&\ddots&\vdots&\vdots\\
0&0&\cdots&P^{\cE}_{1,i_{t_1}}&0\\
P^{\cE}_{2,i_1}&P^{\cE}_{2,i_2}&\cdots&P^{\cE}_{2,i_{t_1}}&P^{\cE}_{2,\ell+1}\\
\end{matrix}},
\end{equation*}
by deleting the all zero rows.  For the case $t_1=0$,
$\rank(P^*(\cE))=\rank(P^{\cE}_{2,\ell+1})=|E_{\ell+1}|$ the erasure
pattern can be recovered.  For the case $t_1\geq 1$, the fact that
$(P^*_{2,1},P^*_{2,2},\dots, P^*_{2,\ell+1})=(p_{u,j})$ is an $h\times
n$ Cauchy matrix with $p_{u,j}=\frac{1}{b_{2,u}-\gamma_{j}}$ for
$1\leq u\leq h$ and $1\leq j\leq n$ means that
\begin{equation*}
\begin{split}
\rank(P^*(\cE))=&\rank\parenv{
\begin{matrix}
P^{\cE}_{1,i_1}&0&\cdots&0&0\\
0&P^{\cE}_{1,i_2}&\cdots&0&0\\
\vdots&\vdots&\ddots&\vdots&\vdots\\
0&0&\cdots&P^{\cE}_{1,i_{t_1}}&0\\
P^{\cE}_{2,i_1}&P^{\cE}_{2,i_2}&\cdots&P^{\cE}_{2,i_{t_1}}&P^{\cE}_{2,\ell+1}\\
\end{matrix}}\\
\geq &\rank
\parenv{\begin{matrix}
P^{\cE}_{1,i_1}&0&\cdots&0&0\\
0&P^{\cE}_{1,i_2}&\cdots&0&0\\
\vdots&\vdots&\ddots&\vdots&\vdots\\
0&0&\cdots&P^{\cE}_{1,i_{t_1}}&0\\
P^{\cE,h_1}_{2,i_1}&P^{\cE,h_1}_{2,i_2}&\cdots&P^{\cE,h_1}_{2,i_{t_1}}&0\\
0&0&\cdots&0&I_{|E_{\ell+1}|}\\
\end{matrix}},
\end{split}
\end{equation*}
where $h_1=h-|E_{\ell+1}|$ and $P^{\cE,h_1}_{2,i_j}$ is the sub-matrix formed by
the first $h_1$ rows of $P^{\cE}_{2,i_j}$ for $1\leq j\leq t_1$.

Recall that $e=\sum_{E\in \cE}|E|\leq \min\{(t+1)\delta-1,h+\delta-1\}$, which means
\begin{equation}\label{eqn_value_t_1}
t_1\leq
\begin{cases}
\floorenv{\frac{(t+1)\delta-1-|E_{\ell+1}|}{\delta}}\leq t, &\text{if }h_1+\delta-1\geq (t+1)\delta-1-|E_{\ell+1}|,\\
\floorenv{\frac{h_1+\delta-1}{\delta}}<t, &\text{if }h_1+\delta-1<(t+1)\delta-1-|E_{\ell+1}|.\\
\end{cases}
\end{equation}
According to \eqref{eqn_cond_S_i}, for $i\in \{i_j~:~1\leq j\leq t_1\}$
$$\left|E_i\cap \left(\bigcup_{\substack{1\leq j\leq t_1\\ i_j\ne i}} E_{i_j}\right)\right|\leq \left|S_i\cap \left(\bigcup_{\substack{1\leq j\leq t_1\\ i_j\ne i}} S_{i_j}\right)\right|\leq \delta-1,$$
which means that the elements of each $E_{i}$ may be indexed $E_{i}=\{\alpha_{i,u} ~:~ 1\leq u\leq \tau_i\}$ such that
\begin{equation*}
\{\alpha_{i,t} ~:~ \delta\leq t\leq \tau_i\}\cap E_{i_j}=\varnothing\text{ for }1\leq j\leq t_1,\,\,i_j\ne i.
\end{equation*}
For $1\leq j\leq t_1$, let $\cE_{i_j}=\{\{\alpha_{i_j,1},\dots,\alpha_{i_j,\delta-1},\alpha_{i_j,u}\}~:~\delta\leq u\leq \tau_{i_j}\}$ and
$\cE^*=\bigcup_{1\leq j\leq t_1}\cE_{i_j}$.
By Lemma~\ref{lemma_M_dia},
$$\rank(P^*(\cE))\geq t_1(\delta-1)+|E_{\ell+1}|+\rank(M(\cE^*,\Theta_3)),$$
where $\Theta_3=\{\gamma_{2,i}~:~1\leq i\leq h_1\}.$
Thus, by Lemma~\ref{lemma_Rank_M}, \eqref{eqn_cond_S_i}, \eqref{eqn_cond_ell}, and \eqref{eqn_value_t_1},
$M(\cE^*,\Theta_3)$ has full rank,
i.e., $\rank(P^*(\cE))\geq |E_{\ell+1}|+\sum_{1\leq i\leq t_1}|E_i|.$
This is to say, the erasure pattern can be recovered, which means
$d\geq \min\{(t+1)\delta,h+\delta\}$.
\end{IEEEproof}

\begin{corollary}\label{corollary_gamma_inf}
Assume $G_1(x)G_2(x)$ has $\delta-1+h$ distinct roots over $\F_{q^{m_1}}$.
Let $\cS$ be a set system of $\F_{q^{m}}$ such that for any $t+1$-subset $D$ of $[\ell]$
\[
\abs{S_i\cap \parenv{\bigcup_{j\ne i,j\in D} S_j}}\leq \delta-1 \quad \text{for }i\in D
\]
and
\[
S_{\ell+1}\cap S_i=\varnothing\quad \text{for }1\leq i\leq \ell.
\]
If $h+\delta\leq (t+1)\delta$ and $S_{\ell+1}\ne \varnothing$, then the
code $\Gamma_{q^m}(\cS,\cG)$ is an optimal $[n,k,d=h+\delta]_{q^m}$ linear code
with $(r,\delta)_i$-locality.
\end{corollary}
\begin{proof}
By Theorem~\ref{theorem_distance}, the facts
$\Gamma_{q^m}(\cS,\cG)\subseteq \Gamma_{q^{m_1}}(\cS,\cG)$ and
$h+\delta\leq (t+1)\delta$ show that $\Gamma_{q^m}(\cS,\cG)$ has minimum
Hamming distance at least $h+\delta$. Thus, by
Lemma~\ref{lemma_Gamma_k}, $\Gamma_{q^m}(\cS,\cG)$ is an $[n,k,d\geq
  h+\delta]_{q^m}$ with $k\geq n-\ell(\delta-1)-h$ and those symbols
with $(r,\delta)$-locality have rank at least
$k_1=n-\ell(\delta-1)-h=\ell r$.  By
Lemma~\ref{lemma_bound_i},
$$d\leq n-k+1-\parenv{\ceilenv{\frac{k_1}{r}}-1}(\delta-1)\leq
n-k+1-(\ell-1)(\delta-1)\leq h+\delta,$$ which together with the fact
$d\geq h+\delta$ show that $d=h+\delta$ and $k=k_1$.  This is also to
say that $\Gamma_{q^m}(\cS,\cG)$ is an optimal linear code with
$(r,\delta)_i$-locality with respect to the bound in
Lemma~\ref{lemma_bound_i}.

\end{proof}

For the case $S_{\ell+1}=\varnothing$, the following corollary follows
directly from Theorem~\ref{theorem_distance} and
Lemma~\ref{lemma_Gamma_k}.

\begin{corollary}
Let $n=\ell(r+\delta-1)$ and $0<h\leq r$.  Assume $G_1(x)G_2(x)$ has
$\delta-1+h$ distinct roots over $\F_{q^{m_1}}$.  Let $\cS$ be a set
system of $\F_{q^{m}}$ such that for any $t+1$-subset $D$ of $[\ell]$
\begin{equation*}
\abs{S_i\cap \parenv{\bigcup_{j\ne i,j\in D} S_j}}\leq \delta-1 \quad \text{for }i\in D
\end{equation*}
If $h+\delta\leq (t+1)\delta$ and $S_{\ell+1}=\varnothing$, then the
code $\Gamma_{q^m}(\cS,\cG)$ is an $[n=\ell(r+\delta-1),k,d\geq
  h+\delta]_{q^m}$ code with $(r,\delta)_a$-locality, where
\begin{equation*}
n-\ell(\delta-1)=\ell r\geq k\geq \ell r-h.
\end{equation*}
Furthermore, if $k=\ell r-h$, then $\Gamma_{q^m}(\cS,\cG)$ is an
$[n,k,d=h+\delta]$ optimal locally repairable code with
$(r,\delta)_a$-locality.
\end{corollary}

\begin{remark}
For the case $k=\ell r-h$ and $h\leq r$, i.e., $S_{\ell+1}=\emptyset$, the codes
generated by Construction \ref{cons_Goppa} share similar parameters
with those constructed in \cite{CMST_SL}. However, Construction
\ref{cons_Goppa} may also work for the case of
$S_{\ell+1}\ne\emptyset$ in which we may construct optimal locally
repairable codes with new parameters as shown in Corollary
\ref{corollary_gamma_inf}.
\end{remark}

\section{Conclusion}\label{sec-conclusion}
In this paper, we first introduced a construction of locally
repairable codes with $(r,\delta)_i$-locality. To analyze the
performance of our construction, an upper bound was derived for the
length of optimal locally repairable codes with
$(r,\delta)_i$-locality. Our main goal, with this bound, is to find a
connection between the length of the code and the field size over
which the code is constructed. Using combinatorial structures
(packings in general, and Steiner systems in particular) we arrive at
the conclusion that, in some cases, the optimal locally-repairable
codes we constructed have order-optimal length, which is super-linear
in the field size. We also suggested another construction for optimal
locally repairable codes, this time, taking inspiration from Goppa
codes. The construction share a similarity in the combinatorial
structures they require. Finally, we defined generalized sector-disk
codes. We showed that the locally repairable codes of our
constructions are capable of yielding GSD codes, and compared their
parameters with sector-disk (SD) codes, and maximally recoverable (MR)
codes.

In general, it seems that constructions of locally repairable codes
have focused mainly on $(r,\delta)_a$-locality, perhaps due to their
symmetry. We believe our constructions and bound show that codes with
$(r,\delta)_i$-locality are also of theoretical and applicative
interest. Several open questions remain, including finding SD/MR/GSD
codes for all possible parameters, and finding more codes that are
capable of correcting special erasure patterns beyond what is
guaranteed due to their Hamming distance.

\appendices

\section{Proof of Theorem~\ref{theorem_bound_delta>2}}
\label{app:proof1}

\begin{lemma}\label{lemma_find_V_1}
  Let $\cC$ be an $[n,k]_q$ linear code with $(r,\delta)_i$-locality
  and $r|k$. Let $\cA$ be the set of all the repair sets of
  information symbols, where we highlight that there may exist some
  information symbols that share the same repair set. For any $1\leq
  j\leq \frac{k}{r}$, if there exists a $j$-subset $\cV\subseteq\cA$
  and $\Delta>0$ is an integer such that for any $A\in \cV$
  \begin{equation}\label{eqn_cond_full_rank}
  \abs{A\cap\parenv{\bigcup_{A'\in \cV\setminus \{A\}}A'}}\leq |A|-\delta+1
  \end{equation}
  and
  \begin{equation*}
    |\cV|(r+\delta-1)-\left|\bigcup_{A\in \cV}A\right|\geq \Delta>0,
  \end{equation*}
 then there exists a set $S\subseteq [n]$ with $\rank(S)=k-1$ and
  \begin{equation*}
    |S|\geq
    k-1+\frac{k}{r}(\delta-1).
  \end{equation*}
\end{lemma}
\begin{IEEEproof}
Let $\cV=\{A_{i_1},A_{i_2},\dots,A_{i_{j}}\}$ and
\begin{equation}\label{eqn_A*}
A^*_{i_t}\subseteq A_{i_t}\setminus \parenv{\bigcup_{A'\in \cV\setminus \{A_{i_t}\}}A'}
\end{equation}
with $|A_{i_t}^*|=\delta-1$ for $1\leq t\leq j$, which is possible in
light of \eqref{eqn_cond_full_rank}. Define a
$(|A_{i_t}|-\delta+1)$-subset $A'_{i_t}\triangleq A_{i_t}\setminus
A^*_{i_t}$ for $1\leq t\leq j$.  By definition~\ref{def_locality}, we
have $\rank(A_{i_t}')=\rank(A_{i_t})$ for $1\leq t\leq j$.
Note that \eqref{eqn_A*} implies that $A^*_{i_t}$ for $1\leq t\leq j$
are pairwise disjoint, which also means that
\begin{equation}\label{eqn_rank_j}
\begin{split}
\rank\parenv{\bigcup_{1\leq t\leq j}A_{i_t}}
=&\rank\parenv{\bigcup_{1\leq t\leq j}(A_{i_t}\setminus A^*_{i_t})}\\
\leq &\abs{\bigcup_{1\leq t\leq j}A_{i_t}\setminus A^*_{i_t}}\\
=&\abs{\bigcup_{1\leq t\leq j}A_{i_t}}-\abs{\bigcup_{1\leq t\leq j}A^*_{i_t}}\\
=&\abs{\bigcup_{1\leq t\leq j}A_{i_t}}-j(\delta-1),
\end{split}
\end{equation}
where the second equality holds by \eqref{eqn_A*}.
This is to
say that
\begin{equation*}
\begin{split}
jr-\rank\parenv{\bigcup_{1\leq t\leq j}A_{i_t}}\geq &j(r+\delta-1)-\abs{\bigcup_{1\leq t\leq j}A_{i_t}}=\Delta>0,
\end{split}
\end{equation*}
i.e., $\rank\parenv{\bigcup_{1\leq t\leq j}A_{i_t}}\leq jr-1$.

 For the case
$j<\frac{k}{r}$, the fact that $\cC$ has $(r,\delta)_i$-locality,
i.e., $\rank\parenv{\bigcup_{A\in \cA}A}=k$ means that there exists an
$A_{i_{j+1}}$ such that $\rank\parenv{\bigcup_{1\leq t\leq
    j+1}A_{i_t}}>\rank\parenv{\bigcup_{1\leq t\leq j+1}A_{i_t}}$.
This is to say $\left|A_{j+1}\cap \parenv{\bigcup_{1\leq t\leq
    j+1}A_{i_t}}\right|\leq |A_{j+1}|-\delta+1$.  Let
$A^*_{i_{j+1}}\subseteq A_{i_{j+1}}\setminus \parenv{\bigcup_{1\leq
    t\leq j+1}A_{i_t}}$ with $|A^*_{i_{j+1}}|=\delta-1$ and
$A'_{i_{j+1}}=A_{i_{j+1}}\setminus A^*_{i_{j+1}}$. Note that
$A_{i_{j+1}}$ is a repair set of $\cC$. Thus,
$\rank(A'_{i_{j+1}})=\rank(A_{i_{j+1}})$ by Definition \ref{def_locality} and
\begin{equation*}
\begin{split}
\rank\parenv{\bigcup_{1\leq t\leq j+1}A_{i_t}}=&\rank\parenv{A'_{i_{j+1}}\cup\parenv{\bigcup_{1\leq t\leq j}A_{i_t}}}\\
\leq& \rank\parenv{\bigcup_{1\leq t\leq j}A_{i_t}}+\abs{A'_{i_{j+1}}\setminus\parenv{\bigcup_{1\leq t\leq j}A_{i_t}}}\\
\leq& \rank\parenv{\bigcup_{1\leq t\leq j}A_{i_t}}+\abs{A_{i_{j+1}}\setminus\parenv{\bigcup_{1\leq t\leq j}A_{i_t}}}-\delta+1\\
=&\abs{\bigcup_{1\leq t\leq j+1}A_{i_t}}-(j+1)(\delta-1),
\end{split}
\end{equation*}
where the last equality holds by \eqref{eqn_rank_j}. Recall that $\rank(\bigcup_{1\leq t\leq j}A_{i_t})<jr$
which means that $\rank(\bigcup_{1\leq t\leq j+1}A_{i_t})<(j+1)r$.

Repeat the preceding analysis $\frac{k}{r}-j$ times, then we can find
$A_{i_t}$ with $1\leq t\leq \frac{k}{r}$ such
that $$\abs{\bigcup_{1\leq
    t\frac{k}{r}}A_{i_t}}-\rank\parenv{\bigcup_{1\leq
    t\frac{k}{r}}A_{i_t}}\geq \frac{k}{r}(\delta-1)$$ and
$\rank\parenv{\bigcup_{1\leq t\frac{k}{r}}A_{i_t}}<k$.  Thus, we can
extend the set $\bigcup_{1\leq t\leq \frac{k}{r}}A_{i_t}$ to be a set
$S$ with $\rank(S)=k-1$ and $$|S|-\rank(S)=|S|-k+1\geq
\abs{\bigcup_{1\leq t\leq
    \frac{k}{r}}A_{i_t}}-\rank\parenv{\bigcup_{1\leq t\leq
    \frac{k}{r}}A_{i_t}}\geq \frac{k}{r}(\delta-1),$$ which means the
desired result follows.
\end{IEEEproof}

\begin{theorem}\label{theorem_repair_sets}
  Let $\cC$ be an optimal $[n,k,d]_q$ linear code with
  $(r,\delta)_i$-locality. If $r|k$ and $r<k$, then there exist
  $\frac{k}{r}$ repair sets
  $\cV=\{A_{i_1},A_{i_2},\dots,A_{i_{\frac{k}{r}}}\}$, such that
  $|A_{i_t}|=r+\delta-1$, $A_{i_t}$ for $1\leq t\leq \frac{k}{r}$ are
  pairwise disjoint and $\rank(\bigcup_{1\leq t\leq
    \frac{k}{r}}A_{i_t})=k$.  Furthermore, the punctured code
  $\cC|_{A_{i_t}}$ for $1\leq t\leq \frac{k}{r}$ is an
  $[r+\delta-1,r,\delta]_q$ MDS code.
\end{theorem}
\begin{IEEEproof}
Since the code $\cC$ has $(r,\delta)_i$-locality, we have
$\rank(\bigcup_{A\in \cA}A)=k$, where $\cA$ denotes the set of all
repair sets of information symbols. Note that for each repair set
$A\in \cA$, by Definition~\ref{def_locality}, we have $\rank(A)\leq
r$.  This means that we can find $A_{i_t}$ for $1\leq t\leq
\frac{k}{r}$ such that $\rank\left(\bigcup_{1\leq t\leq
  j}A_{i_t}\right)>\rank\left(\bigcup_{1\leq t\leq j-1}A_{i_t}\right)$
for $2\leq j\leq \frac{k}{r}$. We claim that those $\frac{k}{r}$
repair sets are pairwise disjoint and $|A_{i_t}|=r+\delta-1$ for
$1\leq t\leq \frac{k}{r}$. Note that for $j>t$ we have $|A_{i_t}\cap
A_{i_j}|\leq |A_j|-\delta+1$, since $\rank\left(\bigcup_{1\leq t\leq
  j}A_{i_t}\right)>\rank\left(\bigcup_{1\leq t\leq
  j-1}A_{i_t}\right)$.  Now by Lemma~\ref{lemma_find_V_1}, if
$2(r+\delta-1)-|A_{i_t}\cup A_{i_j}|>0$ then we have a set $S$ with
rank $k-1$ and $|S|=k-1+\frac{k}{r}(\delta-1)$, which contradicts with
the fact that $\cC$ is optimal, i.e.,
$d=n-k+1-(\frac{k}{r}-1)(\delta-1)$.  Thus, for $j>t$ and $1\leq
j,t\leq \frac{k}{r}$, we have $2(r+\delta-1)-|A_{i_t}\cup A_{i_j}|=0$,
i.e., $A_{i_t}\cap A_{i_j}=\varnothing$ and
$|A_{i_{t}}|=|A_{i_{j}}|=r+\delta-1$, since $|A_{i_{t}}|\leq
r+\delta-1$ and $|A_{i_{j}}|\leq r+\delta-1$.

Now, we only need to prove that $\rank\parenv{\bigcup_{1\leq t\leq
    \frac{k}{r}}A_{i_t}}=k$.  If that is not the case, then we have
$\rank\parenv{\bigcup_{1\leq t\leq \frac{k}{r}}A_{i_t}}\leq k-1$.
Note that $\abs{\bigcup_{1\leq t\leq
    \frac{k}{r}}A_{i_t}}=k+\frac{k}{r}(\delta-1)$, which is also a
contradiction with $d=n-k+1-(\frac{k}{r}-1)(\delta-1)$.  Therefore,
the desired result follows. Finally, for $1\leq t\leq \frac{k}{r}$,
the fact that $\rank(A_{i_t})=r$, $|A_{i_t}|=r+\delta-1$, and
$d(\cC|_{A_{i_t}})\geq \delta$, shows that $\cC|_{A_{i_t}}$ is an
$[r+\delta-1,r,\delta]_q$ MDS code.
\end{IEEEproof}

We are now in a position to prove Theorem~\ref{theorem_bound_delta>2}.
\begin{IEEEproof}
By Theorem~\ref{theorem_repair_sets}, and up to a rearrangement of the code
coordinates, the parity-check matrix $P$ of code $\cC$ can be arranged
in the following form,
\begin{equation*}
P=\begin{pmatrix}
      L^{(1)} & 0 & 0 & \dots & 0 &0\\
      0 &  L^{(2)} & 0 & \dots & 0 &0\\
      0 & 0 & L^{(3)} & \dots & 0 &0\\
      \vdots & \vdots & \vdots & \ddots & \vdots &\vdots \\
      0 & 0 & 0 & \dots & L^{(\ell)}&0 \\
      H_1 & H_2 & H_3 & \dots & H_\ell & H_{\ell+1}\\
    \end{pmatrix},
\end{equation*}
where $L^{(i)}=(I_{\delta-1}, P_i)$ is a $(\delta-1)\times
(r+\delta-1)$ matrix for all $1\leq i\leq w$ and we do row linear
transformations to make sure each $L^{(i)}$ has canonical form.
Define
\begin{equation}\label{eqn_matrix_p_1}
M_1\triangleq \begin{pmatrix}
I_{\delta-1}&0 & 0 &0& \dots & 0 &0& 0 \\
0 &0&  I_{\delta-1} & 0 & \dots & 0& 0 & 0\\
0 & 0&0 & 0&\dots & 0& 0&0 \\
\vdots&\vdots &\vdots& \vdots &  &\vdots& \vdots&\vdots \\
0 &0& 0 & 0 & \dots &I_{\delta-1}&0&0 \\
0&H^{(1)}_{1} &0& H^{(1)}_{2} & \dots & 0 &H^{(1)}_{\ell}& H_{\ell+1} \\
\end{pmatrix},
\end{equation}
where $I_{\delta-1}$ denotes the $(\delta-1)\times(\delta-1)$ identity
matrix and $H^{(1)}_i=H_{i,2}-H_{i,1}P_i$ with $H_i=(H_{i,1},H_{i,2})$
and $1\leq i\leq \ell$.  For any integer $0\leq a\leq h$,
let \begin{equation*}
M_{2,a}=
\begin{pmatrix}
  H^{(1)}_1 & H^{(1)}_2 & H^{(1)}_3 & \dots &H^{(1)}_{\ell}& H^{(a)}_{\ell+1} \\
\end{pmatrix},
\end{equation*}
where $H^{(a)}_{\ell+1}$ denotes the matrix generated by deleting any $a$ columns
from $H_{\ell+1}$.

Now, for any $0\leq a\leq h$, the fact that any $d-1$ columns of $P$
are linearly independent over $\F_q$ means that any
$T(a)=\lfloor\frac{d-a-1}{\delta}\rfloor$ columns of $M_2$ are
linearly independent over $\F_q$. This is because any $T(a)$ columns
of $M_{2,a}$ correspond to at most $T(a)\delta$ columns of $P$ by
adding the first $\delta-1$ columns in related blocks, and by
\eqref{eqn_matrix_p_1} they have full column rank.  Therefore,
$M_{2,a}$ is the parity-check matrix of a linear code $\cC_{1,a}$,
with parameters $[\ell r+h-a,k'\geq k=\ell r,d_2\geq {T(a)}+1]_q$.

In what follows, we distinguish between two cases, depending on the
parity of $T(a)$.

Case 1: $T(a)$ is odd.  In this case, we consider the shortened code
$\cC_{2,a}$ of $\cC_{1,a}$ with parameters $[\ell r+h-a-1,k'\geq \ell
  r,d_2\geq t]_q$. By the Hamming bound \cite{MS} we have
\begin{equation*}
q^{\ell r}\leq \frac{q^{\ell r+h-a-1}}{\sum_{0\leq i\leq \frac{T(a)-1}{2}}\binom{\ell r+h-a-1}{i}(q-1)^i}\leq \frac{q^{\ell r+h-a-1}}{\binom{\ell r+h-1}{\frac{T(a)-1}{2}}(q-1)^{\frac{T(a)-1}{2}}}\leq \frac{q^{\ell r+h-a-1}}{\left(\frac{\ell r+h-a-1}{\frac{T(a)-1}{2}}\right)^{\frac{T(a)-1}{2}}(q-1)^{\frac{T(a)-1}{2}}},
\end{equation*}
which means
\begin{equation*}
\ell r+h-a-1\leq \frac{T(a)-1}{2(q-1)}q^{\frac{2(h-a-1)}{T(a)-1}}.
\end{equation*}
This is to say,
\begin{equation*}
n\leq \frac{r+\delta-1}{r}\parenv{\frac{T(a)-1}{2(q-1)}q^{\frac{2(h-a-1)}{T(a)-1}}-h+a+1}+h
=\frac{r+\delta-1}{r}\parenv{\frac{T(a)-1}{2(q-1)}q^{\frac{2(h-a-1)}{T(a)-1}}+a+1}-\frac{h(\delta-1)}{r}.
\end{equation*}

Case 2: $T(a)$ is even.  Similarly, by the Hamming bound, we have
\begin{equation*}
q^{\ell r}\leq \frac{q^{\ell r+h-a}}{\sum_{1\leq i\leq \frac{T(a)}{2}}\binom{\ell r+h-a}{i}(q-1)^i}\leq \frac{q^{\ell r+h-a}}{\binom{\ell r+h-a}{\frac{T(a)}{2}}(q-1)^{\frac{T(a)}{2}}}\leq \frac{q^{\ell r+h-a}}{\left(\frac{\ell r+h-a}{\frac{T(a)}{2}}\right)^{\frac{T(a)}{2}}(q-1)^{\frac{T(a)}{2}}},
\end{equation*}
which means
$$n=\ell(r+\delta-1)+h\leq \frac{r+\delta-1}{r}\parenv{\frac{T(a)}{2(q-1)}q^{\frac{2(h-a)}{T(a)}}+a}-\frac{h(\delta-1)}{r}.$$
Finally, recall that by Lemma~\ref{lemma_bound_i}, $\mathcal{C}$ is
optimal means that $h=d-\delta$. This completes the proof.
\end{IEEEproof}

\section{Regular Packings}
\label{app:pack}

We present a direct construction of regular packings based on a kind
of cyclotomy. The generated regular packings are not new, and may
obtained recursively via \cite{J}, and via generalized cyclotomy
\cite{DWX,ZCTY2013}. Thus, the construction and proof herein are
brought for the reader's convenience only.

According to the unique factorization theorem, a positive integer $n$ has the following unique decomposition
\begin{eqnarray*}\label{eqn-facto}
n = p_1^{m_1}p_2^{m_2}\cdots p_u^{m_u},
\end{eqnarray*}
where $p_1<p_2<\cdots<p_u$ are primes and $m_1,m_2,\ldots,m_u$ are positive integers.
For $1\leq i\leq u$, let $\F_{p^{m_i}_i}$ be the finite field with size $p^{m_i}_i$ and
$\alpha_i$ be one of its primitive elements. Let $e$ be a positive integer with
\begin{equation*}
e\mid\gcd(p^{m_1}_1-1,p^{m_2}_2-1,\cdots,p^{m_u}_u-1).
\end{equation*}
For $e>1$, define
\begin{equation*}
{\bm \beta_e}\triangleq(\alpha_1^{\frac{p^{m_1}-1}{e}},\alpha_2^{\frac{p^{m_2}-2}{e}},\dots,\alpha_u^{\frac{p^{m_u}-1}{e}})\in T\triangleq \F_{p^{m_1}_{1}}\times\F_{p^{m_2}_{2}}\times\dots \times\F_{p^{m_u}_{u}}.
\end{equation*}
It is easy to verify that $D_{0}=\langle{\bm \beta_e}\rangle=\{{\bm \beta^0_e},{\bm \beta^1_e},\cdots,{\bm \beta^{e-1}_e}\}\subseteq T^*\triangleq\F^*_{p^{m_1}_{1}}\times \F^*_{p^{m_2}_{2}}\times\dots \times\F^*_{p^{m_u}_{u}}$ is a subgroup of $(T^*,\cdot)$
with order $e$, where $${\bm \beta^i_e}\triangleq(\alpha_1^{i\frac{p^{m_1}-1}{e}},\alpha_2^{i\frac{p^{m_2}-2}{e}},\dots,\alpha_u^{i\frac{p^{m_u}-1}{e}})\in T.$$

For ${J}\in A\triangleq\Z_{\frac{p^{m_1}-1}{e}}\times \Z_{\frac{p^{m_1}-1}{e}}\times \dots\Z_{\frac{p^{m_u}-1}{e}}$,
define $B_{J}\subseteq \Z_e\times T$ as
\begin{equation}\label{eqn_B_j}
B_{J}\triangleq \{(0,{\bm \alpha^{J}}{\bm \beta_e^0}),(1,{\bm \alpha^{J}}{\bm \beta_e^1}),\dots,(e-1,{\bm \alpha^{J}}{\bm \beta_e^{e-1}})\},
\end{equation}
where $\bm \alpha\triangleq(\alpha_1,\alpha_2,\dots,\alpha_u)$ and $\bm \alpha^{J}\triangleq(\alpha^{j_1}_1,\alpha^{j_2}_2,\dots,\alpha^{j_u}_u)$ for $J=(j_1,j_2,\dots,j_u)$.
Based on $B_{J}$s, we can generate a set system as:
\begin{construction}\label{cons_packings}
Let $X=\Z_e\times T$,
then we may construct a set
\begin{equation}\label{eqn_B_j_eps}
\cB=\{B_{{J},\epsilon}=B_{J}+(0,\epsilon)~:~{J}\in A, \epsilon\in T\}.
\end{equation}
\end{construction}

\begin{theorem}
The set system $(X,\cB)$ generated by Construction~\ref{cons_packings} is a
$\frac{\prod_{1\leq i\leq u} (p^{m_i}_1-1)}{e^u}$-regular packing with parameters $(en,e,1)$.
\end{theorem}
\begin{IEEEproof}
By Construction~\ref{cons_packings}, it is sufficient to prove that
any pair of elements of $X$ appears in at most one of the blocks in
$\cB$.  Assume to the contrary that there exists a pair
$\{x_1=(i_1,\gamma_1), x_2=(i_2,\gamma_2)\}\subseteq X$ that appears
in two blocks, i.e., $\{x_1,x_2\}\subseteq B_{J_1,\epsilon_1}$ and
$\{x_1,x_2\}\subseteq B_{J_2,\epsilon_2},$ where $i_1,i_2\in \Z_{e}$ and
$\gamma_1,\gamma_2\in T.$ By \eqref{eqn_B_j} and \eqref{eqn_B_j_eps},
there exist four elements $t_{1,1},t_{1,2},t_{2,1},t_{2,2}\in \Z_e$
such that
\begin{equation}\label{eqn_cond1}
(0,\epsilon_1)+(t_{1,1},{\bm\alpha^{J_1}}{{\bm \beta}^{t_{1,1}}_e})=(i_1,\gamma_1)=(0,\epsilon_2)+(t_{2,1},{\bm\alpha^{J_2}}{{\bm \beta}^{t_{2,1}}_e})
\end{equation}
and
\begin{equation}\label{eqn_cond2}
(0,\epsilon_1)+(t_{1,2},{\bm\alpha^{J_1}}{{\bm \beta}^{t_{1,2}}_e})=(i_2,\gamma_2)=(0,\epsilon_2)+(t_{2,2},{\bm\alpha^{J_2}}{{\bm \beta}^{t_{2,2}}_e}).
\end{equation}
These equalities imply that $t_{1,1}=t_{2,1}$, $t_{1,2}=t_{2,2}$ and
$${\bm\alpha^{J_1}}{\bm \beta}_e^{t_{1,2}-t_{1,1}}={\bm\alpha^{J_2}}{\bm \beta}^{t_{2,2}-t_{2,1}}_e,$$
i.e.,
\begin{equation}\label{eqn_alpha_J_12}
{\bm\alpha^{J_1}}={\bm\alpha^{J_2}}.
\end{equation} Note that ${J_1},{J_2}\in A=\Z_{\frac{p^{m_1}-1}{e}}\times \Z_{\frac{p^{m_1}-1}{e}}\times \dots\Z_{\frac{p^{m_u}-u}{e}}$
and $\bm \alpha=(\alpha_1,\alpha_2,\cdots, \alpha_u)$, where
$\alpha_i$ is a primitive element of $\F_{p^{m_i}_i}$.  Thus, by
\eqref{eqn_alpha_J_12}, we have $J_1=J_2$. Again by
\eqref{eqn_cond1} and \eqref{eqn_cond2}, we have
$\epsilon_1=\epsilon_2$, a contradiction. Thus, the desired result
follows.
\end{IEEEproof}


\begin{thebibliography}{1}


\bibitem{B} S. Ball, ``On sets of vectors of a finite vector space in which every subset of basis size is a basis,"
\emph{J. Eur. Math. Soc.}, vol. 14, no .3, pp. 733--748, 2012.

\bibitem{BHP} R. C. Baker, G. Harman, and J. Pintz, ``The difference between consecutive primes, II," \emph{Proceedings of the London Math. Soc.,}
vol. 83, no. 3, pp. 532--562, Nov. 2001.

\bibitem{Bla} M. Blaum, ``Extended integrated interleaved codes over any field with
applications to locally recoverable codes" to appear in \emph{IEEE Trans. Inf. Theory}.

\bibitem{BHH} M. Blaum, J. L. Hafner, and S. Hetzler, ``Partial-MDS codes and their
application to RAID type of architectures," \emph{IEEE Trans. Inf. Theory,}
vol. 59, no. 7, pp. 4510--4519, Jul. 2013.

\bibitem{BH} M. Blaum and S. R. Hetzler, ``Array Codes with Local Properties," \emph{arXiv:1906.11731}, 2019.


\bibitem{BPSY}M. Blaum, J. Plank, M. Schwartz, and E. Yaakobi, ``Construction of partial MDS
and sector-disk codes with two global parity symbols," \emph{IEEE Trans. Inf. Theory,}
vol. 62, no. 5, pp. 2673--2681, May 2016.

\bibitem{BCGLP} A. Beemer, R. Coatney, V. Guruswami, H. H. L\'{o}pez, and  F. Pi\~{n}ero, ``Explicit optimal-length locally repairable codes of distance 5," arXiv:1810.03980, 2018.

\bibitem{CK} G. Calis and O. O. Koyluoglu, ``A general construction for PMDS codes," \emph{IEEE Commun. Lett.,} vol. 21, no. 3, pp. 452--455, Mar. 2017.

\bibitem{CM} V. R. Cadambe and A. Mazumdar, ``Bounds on the size of locally recoverable codes," \emph{IEEE Trans. Inf. Theory,}
vol. 61, no. 11, pp. 5787--5794, Nov. 2015.

\bibitem{CCFT} H. Cai, M. Cheng, C. Fan, and X. Tang, ``Optimal locally repairable systematic codes based on packings," \emph{IEEE Trans. Commun.},
vol. 67, no .1, pp. 39--49, Jan. 2019.

\bibitem{CMST_SL} H. Cai, Y. Miao, M. Schwartz, and X. Tang, ``On optimal locally repairable codes with super-linear length,"
to appear in \emph{IEEE Trans. Inf. Theory}.

\bibitem{CMST} H. Cai, Y. Miao, M. Schwartz, and X. Tang, ``On optimal locally repairable codes with multiple disjoint repair sets,"
to appear in \emph{IEEE Trans. Inf. Theory}.

\bibitem{CXHF} B. Chen, S. Xia, J. Hao, and F. Fu, ``Constructions of optimal cyclic $({r},{\delta})$ locally repairable codes,"
 \emph{IEEE Trans. Inf. Theory}, vol. 64, no. 4, pp. 2499-2511, Apr. 2017.


\bibitem{CD} C. J.  Colbourn and J. H. Dinitz, \emph{Handbook of Combinatorial Designs}, Chapman
\& Hall/CRC, vol. 42, 2006.



\bibitem{DWX} C. Ding, Q. Wang, and M. Xiong, ``Three new families of zero-difference balanced functions with applications," \emph{IEEE Trans. Inf. Theory,} vol. 60, no. 4, pp. 2407--2413, Apr. 2014.

\bibitem{D} C. Ding, \emph{Designs from Linear Codes}, Singapore: World Scientific, 2018.

\bibitem{GYBS}R. Gabrys, E. Yaakobi, M. Blaum, and P. H. Siegel, ``Constructions of partial MDS codes over small fields," \emph{In Proc. of IEEE ISIT}, 2017.


\bibitem{GHJY} P. Gopalan, C. Huang, B. Jenkins, and S. Yekhanin, ``Explicit maximally recoverable codes with locality," \emph{IEEE Trans. Inf. Theory,} vol. 60, no.9, pp.
5245--5256, Sept. 2014.


%
\bibitem{GHSY} P. Gopalan, C. Huang, H. Simitci, and S. Yekhanin, ``On the locality
of codeword symbols," \emph{IEEE Trans. Inf. Theory,} vol. 58, no. 11, pp.
6925--6934, Nov. 2012.

\bibitem{GGY} S. Gopi, V. Guruswami, and S. Yekhanin. ``On maximally recoverable local reconstruction codes,"  \emph{arXiv:1710.10322}, 2017.

\bibitem{G70} V. D. Goppa,  ``A new class of linear correcting codes," \emph{Problems of Inform. Trans.,} vol. 6, no. 3, pp. 207--212, 1970.


\bibitem{GXY} V. Guruswami, C. Xing, and C. Yuan, ``How long can optimal locally repairable codes be?"
\emph{IEEE Trans. Inf. Theory,} vol. 65, no. 6, pp. 3662--3670, Jun. 2019.



\bibitem{HX} J. Hao and S. Xia, ``Constructions of optimal binary locally repairable codes with multiple repair groups," \emph{IEEE Commun. Lett.}, vol. 20, no. 6, pp. 1060--1063, Jun. 2016.


\bibitem{HCL} C. Huang, M. Chen, and J. Li, ``Pyramid codes: Flexible schemes to
trade space for access efficiency in reliable data storage systems," In NCA, 2007.
%

\bibitem{Jin} L. Jin, ``Explicit construction of optimal locally recoverable codes of distance $5$ and $6$ via binary constant weight codes,"
\emph{IEEE Trans. Inf. Theory,} vol. 65, no. 8, pp. 4658--4663, Aug. 2019.


\bibitem{J} D. Jungnickel, ``Composition theorems for difference families and regular planes," \emph{Discr. Math.,}
vol. 23, no. 2, pp. 151--158, 1978.

\bibitem{K} P. Keevash, ``The existence of designs," arXiv:1401.3665v1, 2014.

\bibitem{KL} G. Kim and J. Lee, ``Locally repairable codes with unequal locality requirements," \emph{IEEE Trans. Inf. Theory,} vol. 64, no. 11, pp. 7137--7152, Nov. 2018.


\bibitem{LL} M. Li and P. P. C. Lee, ``Stair codes: A general family of erasure
codes for tolerating device and sector failures," \emph{ACM Transactions on Storage,}
vol. 10, no. 4, article 14, 2014.

\bibitem{LMX} X. Li, L. Ma, and C. Xing, ``Optimal locally repairable codes via elliptic curves,"  \emph{IEEE Trans. Inf. Theory,} vol. 65, no. 1, pp. 108-117, Jan. 2018.

\bibitem{LMC} J. Liu , S. Mesnager, and L. Chen, ``New constructions of optimal locally recoverable codes via good polynomials," \emph{IEEE Trans. Inf. Theory,} vol. 64, no. 2, pp. 889--899, Feb. 2018.


\bibitem{MS} F. J. MacWilliams and N. J. A. Sloane,  \emph{The Theory of Error-Correcting Codes}, North-Holland, 1977.

\bibitem{MK} U. Mart\'{i}nez-Pe\~{n}as and F. R. Kschischang, ``Universal and dynamic locally repairable codes with maximal recoverability via sum-rank codes,"
\emph{IEEE Trans. Inf. Theory,} vol. 65, no. 12, pp. 7790-7805, Dec 2019.

\bibitem{PB} J. S. Plank and M. Blaum, ``Sector-disk (SD) erasure codes for mixed failure modes in RAID systems,"
\emph{ACM Transactions on Storage,} vol. 10, no. 1, article 4, 2014.

%

%

%
%
\bibitem{PKLK} N. Prakash, G. M. Kamath, V. Lalitha, and P. V. Kumar, ``Optimal linear
codes with a local-error-correction property," \emph{In Proc.
of IEEE ISIT}, 2012.
%
%


\bibitem{RKSV} A. S. Rawat, O. O. Koyluoglu, N. Silberstein, and S. Vishwanath,
``Optimal locally repairable and secure codes for distributed storage
systems," \emph{IEEE Trans. Inf. Theory,} vol. 60, no. 1, pp. 212--236, 2014.
%
%
\bibitem{RPDV} A. S. Rawat, D. S. Papailopoulos, A. G. Dimakis, and S. Vishwanath, ``Locality and availability in distributed storage," \emph{IEEE Trans. Inf. Theory,} vol. 62, no. 8, pp. 4481--4493, 2016.

\bibitem{R} V. R\"{o}dl, ``On a packing and covering problem," \emph{European Journal of Combinatorics} vol. 6, no. 1, pp. 69--78, 1985.


\bibitem{SAK} B. Sasidharan, G. K. Agarwal, and P. V. Kumar, ``Codes with hierarchical locality," \emph{In Proc. of IEEE ISIT}, 2015.


\bibitem{SAPDVCB}M. Sathiamoorthy, M. Asteris, D. Papailiopoulos, A. G. Dimakis
  R. Vadali, S. Chen, and D. Borthakur "Xoring elephants: Novel erasure codes for big data," \emph{Proceedings of the VLDB Endowment,} vol. 6. no. 5. VLDB Endowment, 2013.

\bibitem{SES} N. Silberstein, T. Etzion, and M. Schwartz, ``Locality and availability of array codes constructed from subspaces,"
\emph{IEEE Trans. Inf. Theory,} vol. 65, no. 5, pp. 2648--2660, May 2019.


\bibitem{SDYL} W. Song, S.H. Dau, C. Yuen, and T.J. Li, ``Optimal locally  repairable linear codes," \emph{IEEE J. Slect. Areas Commun.,} vol. 32, no. 5, pp.
1019--1036, May 2014.

\bibitem{TB} I. Tamo and A. Barg, ``A family of optimal locally recoverable codes,"
\emph{IEEE Trans. Inf. Theory,} vol. 60, no. 8, pp. 4661--4676, 2014.

\bibitem{TBF} I. Tamo, A. Barg, and A. Frolov, ``Bounds on the parameters of locally recoverable codes," \emph{IEEE Trans. Inf. Theory,} vol. 62, no. 6, pp. 3070--3083, Jun. 2016.

\bibitem{TPD} I. Tamo, D. Papailiopoulos, and A.G. Dimakis, ``Optimal locally repairable codes and connections to matroid theory," \emph{IEEE Trans. Inf. Theory,} vol. 62, no. 12, pp. 6661--6671, Dec. 2016.

\bibitem{WZ} A. Wang, and Z. Zhang, ``Repair locality with multiple erasure tolerance," \emph{IEEE Trans. Inf. Theory,} vol. 60, no. 11, pp.
6979--69878, Nov. 2014

\bibitem{WFEH} T. Westerb\"{a}ck, R. Freij-Hollanti, T. Ernvall, and C. Hollanti, ``On the Combinatorics of locally repairable codes via matroid theory," \emph{IEEE Trans. Inf. Theory,} vol. 62, no. 10, pp. 5296--5315, Oct. 2016.

\bibitem{ZY} A. Zeh and E. Yaakobi, ``Bounds and constructions of codes with multiple localities," \emph{In Proc. of IEEE ISIT}, 2016.

\bibitem{ZCTY2013} X. Zeng, H. Cai, X. Tang, and Y. Yang, ``Optimal frequency hopping sequences of odd length," \emph{IEEE Trans. Inf. Theory,} vol. 59, no. 5, pp. 3237--3248, May 2013.



\end{thebibliography}
\end{document}